\documentclass[sigconf]{acmart}

\usepackage[utf8]{inputenc}
\usepackage{caption}
\usepackage{subcaption}
\usepackage{url}
\usepackage{float}
\usepackage{amsmath}
\usepackage{hyperref}
\usepackage{graphicx}
\usepackage{markdown}
\usepackage{amsthm}
\usepackage{appendix}
\usepackage{xcolor}
\usepackage{enumitem}
\usepackage{tikz}
\usetikzlibrary{arrows}
\usetikzlibrary{patterns}
\usetikzlibrary{positioning}
\usepackage{multirow}

\newcommand{\ignore}[1]{}

\newcounter{revcount}
\newcommand{\revhighlight}[2][]{%
\refstepcounter{revcount}%
 \ifodd\displayrev%
  {\color{blue}{[\therevcount~HIGHLIGHT, #1:] #2}}%
 \else%
  #2%
 \fi%
}

\newcommand{\revdelete}[2][]{%
\refstepcounter{revcount}%
\ifodd\displayrev%
  {\color{red}{[\therevcount~DELETED, #1:] #2}}%
 \else%
 \fi%
}

\newcommand{\revadd}[2][]{%
\refstepcounter{revcount}%
\ifodd\displayrev%
  {\color{teal}{[\therevcount~ADDED, #1:] #2}}%
 \else%
  #2%
 \fi%
}

\newcommand{\revreplace}[3][]{%
\refstepcounter{revcount}%
 \ifodd\displayrev%
  {\color{red}{[\therevcount~REPLACED, #1:] #2}}\ {\color{teal}{[BY:] #3}}%
 \else%
  #3%
 \fi%
}

\newtheorem{defn}{Definition}
\newtheorem{cor}{Corollary}

\newtheorem{prop}{Proposition}
\newtheorem{nota}{Notation}

\newtheorem{assu}{Assumption}

\newcommand{\mf}[1]{{\color{red}{ [MF: #1]}}}
\newcommand{\xw}[1]{{\color{blue}{ [XW: #1]}}}
\renewcommand{\H}{{\mathcal{H}}}
\newcommand{\Z}{{\mathcal{Z}}}
\newcommand{\A}{{\mathcal{A}}}
\newcommand{\B}{{\mathcal{B}}}

\newcommand{\chaincount}{\textbf{DROP chain count}}
\newcommand{\ProtPCpos}{\textbf{ADAPT PC pos}}
\newcommand{\chaintuple}[1]{\left(#1\right)_{\textbf{ADAPT chain tuple}}}
\newcommand{\chainindex}{\textbf{DROP chain index}}

\newcommand{\bits}{\{0,1\}}

\newcommand{\genesis}{{\cal G}}
\newcommand{\chain}{\mathsf{C}}
\newcommand{\Chains}{\mathcal{C}}
\newcommand{\chainset}{\mathbb{C}}
\newcommand{\Blocks}{\mathcal{B}}
\newcommand{\blockset}{\mathbb{B}}
\newcommand{\conblockset}{\mathbb{V}}
\newcommand{\head}{\mathsf{head}}
\newcommand{\past}{\mathsf{past}}

\newcommand{\eqdef}{\stackrel{\triangle}{=}}
\newcommand{\defeq}{\eqdef}

\newcommand{\party}{\mathsf{P}}
\newcommand{\epochlen}{R}
\newcommand{\slot}{\mathrm{sl}}
\newcommand{\thres}{T}

\newcommand{\VRF}{\mathsf{VRF}}

\newcommand{\vrflengthout}{\ell_{\mathsf{VRF}}}
\newcommand{\vvrf}{v^{\mathrm{vrf}}}

\newcommand{\vrfout}{y}
\newcommand{\vrfproof}{\pi}

\newcommand{\vkes}{v^{\mathrm{kes}}}
\newcommand{\vsig}{v^{\mathrm{sig}}}
\newcommand{\actvsl}{f}
\newcommand{\relstake}{\alpha}
\newcommand{\Stakedist}{\mathbb{S}}
\newcommand{\nonce}{\eta}
\newcommand{\slotnow}{{\slot^*}}
\newcommand{\jnow}{{j^*}}
\newcommand{\data}{d}
\newcommand{\blkproof}{{\pi_B}}
\newcommand{\nonceseed}{\rho}
\newcommand{\kessig}{\sigma}

\newcommand{\msg}[2]{\ensuremath{(\mathsf{#1}, #2)}}
\newcommand{\concat}{\,\|\,}

\newcommand{\functionality}{\mathcal{F}}
\newcommand{\Fdiffuse}{\functionality_{\mathsf{diff}}}
\newcommand{\Finit}{\functionality_{\mathsf{init}}}
\newcommand{\Fvrf}{\functionality_{\mathsf{vrf}}}
\newcommand{\Fsig}{\functionality_{\mathsf{sig}}}
\newcommand{\Fkes}{\functionality_{\mathsf{kes}}}
\newcommand{\RO}{\functionality_{\mathsf{ro}}}

\newcommand{\ro}{\mathsf{H}}
\newcommand{\GetValidTX}{\mathsf{GetValidTX}}
\newcommand{\TruncByEpoch}[2]{#1^{|#2}} 
\newcommand{\Included}{\mathcal{I}}
\newcommand{\maxvalid}{\mathsf{selectchain}}

\newenvironment{mfboxfig}[2]{
  \begin{figure}[t!]
    \newcommand{\FigCaption}{#1}
    \newcommand{\FigLabel}{#2}
    \begin{center}
      \begin{small}
      \setlist{nosep}
        \begin{tabular}{@{}|@{~~}l@{~~}|@{}}
          \hline
          \rule[-1.5ex]{0pt}{1ex}
          \begin{minipage}[b]{0.96\linewidth}
            \mfskip{-1mm}
             \caption{\FigCaption}
             \label{\FigLabel}
             \smallskip
            }{%
          \end{minipage}\\
          \hline
        \end{tabular}
      \end{small}
    \end{center}
    \vskip -8mm
  \end{figure}
}

\newcommand{\env}{\ensuremath{\mathcal{Z}}}

\newcommand{\mfskip}[1]{} 

\def\E{\mathbb{E}}
\def\V{\mathbb{V}}

\newcommand{\minotaur}{{\sf Minotaur} }
\newcommand{\minotaurnosp}{{\sf Minotaur}\ignorespaces}

\newcommand{\bitcoin}{{\sf Bitcoin} }
\newcommand{\bitcoinnosp}{{\sf Bitcoin}\ignorespaces}

\newcommand{\fruitchains}{{\sf FruitChain} }
\newcommand{\fruitchainsnosp}{{\sf FruitChain}\ignorespaces}

\renewcommand\footnotetextcopyrightpermission[1]{} 
\acmYear{2022}
\begin{document}
\fancyhead{}
\def\displayrev{0} 

\title{\minotaurnosp: Multi-Resource Blockchain Consensus
}

\author{Matthias Fitzi} \authornote{The first two authors contributed equally to this work. For correspondence on the paper, please contact Matthias Fitzi at matthias.fitzi@iohk.io.}
\affiliation{%
  \institution{IOG}
  \country{}
}
\email{matthias.fitzi@iohk.io}

\author{Xuechao Wang}\authornotemark[1]
\affiliation{%
  \institution{University of Illinois Urbana-Champaign}
  \country{}
}
\email{xuechao2@illinois.edu}

\author{Sreeram Kannan}
\affiliation{%
  \institution{University of Washington at Seattle}
  \country{}
}
\email{ksreeram@uw.edu}

\author{Aggelos Kiayias}
\affiliation{
  \institution{University of Edinburgh \& IOG}
  \country{}
}
\email{aggelos.kiayias@ed.ac.uk}

\author{Nikos Leonardos}
\affiliation{
  \institution{University of Athens}
  \country{}
}
\email{nikos.leonardos@gmail.com}

\author{Pramod Viswanath}
\affiliation{%
  \institution{Princeton University}
  \country{}
}
\email{pramodv@princeton.edu}

\author{Gerui Wang}
\affiliation{%
  \institution{Beijing Academy of Blockchain and Edge Computing}
  \country{}
}
\email{wanggerui@baec.org.cn}

\begin{abstract}
Resource-based consensus is the backbone of permissionless distributed ledger systems. 
The security of such protocols relies fundamentally on the level of resources actively engaged in the system. 
The variety of different resources (and related proof protocols, some times referred to as PoX in the literature) raises the fundamental question whether it is possible to utilize many of them in tandem and build {\em multi-resource} consensus protocols. The challenge in combining different resources is to achieve {\em fungibility} between them, in the sense that security would hold as long as the {\em  cumulative} adversarial power across all resources is bounded.  

In this work, we put forth \minotaurnosp, a multi-resource blockchain consensus protocol that combines proof-of-work (PoW) and proof-of-stake (PoS), and we prove it {\em  optimally} fungible. 
At the core of our design, \minotaur operates in epochs while continuously sampling the active computational power to provide a fair exchange between the two resources, work and stake. Further, we demonstrate the ability of \minotaur to handle a higher degree of work fluctuation as compared to the \bitcoin blockchain; we also generalize \minotaur to any number of resources.  
  
We demonstrate the simplicity of \minotaur via implementing a full stack client in Rust (available open source~\cite{Minotaurcode}). We use the client to test the robustness of \minotaur to variable mining power and combined work/stake attacks and demonstrate concrete empirical evidence towards the suitability of \minotaur to serve as the consensus layer of a real-world blockchain. 
\end{abstract}

\maketitle
\pagestyle{plain}

\section{Introduction}\label{sec:intro}








\noindent {\bf Resource-based Consensus}. The fundamental feature of the decentralized computing paradigm of permissionless blockchains is that participation in the consensus protocol is enabled by proving possession of a resource.
\bitcoin pioneered this idea based on proof-of-work (PoW), and its implied energy wastage inspired new designs based on alternative resources such as proof-of-stake (PoS), proof-of-space (PoSp) and proof-of-elapsed-time (PoD).
These different resources cover the variegated and multidimensional forms of `capital' that the participants bring.
Each  proof-of-X mechanism is interesting on their merit in appropriate settings (computation in PoW, memory and storage in PoSp and time/delay in PoD, tokens/capital in PoS).
Further, the mechanisms trigger different incentives even within the same resource format; focusing on PoW for example, we see that the hashcash algorithm~\cite{back2002hashcash} instantiated by the SHA256 function~\cite{sha256} implemented in \bitcoin has been replaced by others,
including scrypt~\cite{percival2016scrypt} and ethash~\cite{ethash}.
Each of these different proof-of-X mechanisms have led to different, and isolated, blockchain ecosystems. 

\noindent {\bf Combining different resources}. Given the diversity of incentivization embodied by  different resources, it is a natural question whether it is feasible to combine their features into a single blockchain design (e.g., hybrid PoW-PoS permissionless blockchains).
The key challenge to combining multiple resources is in determining the {\em exchange rate}, i.e., to what extent the different resources are translated into power of authority in the system
and to extract security from these resources in an optimal manner. 
Adapting the exchange rate dynamically to the participation levels in each resource type, resulting in a truly {\em fungible} notion of resource, is a basic and fundamental goal. 
By truly fungible, we mean that the security of the multi-resource consensus protocol is guaranteed as long as the honest players control a majority of the combined resources that consist of each resource type in the system --- concretely, 
$\sum^M_{i=1} \beta_i <M/2$ where $M$ is the number of resources
and $\beta_i$ the adversarial power in the $i$-th resource
(this property is appropriately generalized to any linear combination in  Definition~\ref{def:adv}).  
\revadd[Request~\ref{req:motivation}\label{ans:motivation}]{We note that achieving this type of fungibility is beneficial for the security of the underlying blockchain system as the adversary will be unable to launch a successful attack by commanding merely $51\%$ in one of the underlying resources.}
Next we give some examples of {\em non-fungible} hybrid PoW/PoS protocol designs.

\noindent {\bf First order hybrid PoW/PoS protocols.} Given the basic importance of incorporating multiple resources in a single blockchain design, there are several  designs of  hybrid PoW/PoS protocols in the literature~\cite{KinNad12,BLMR14,FDKT+20,DCFZ18,karakostas2021securing}. These  protocols, however, are either heuristic (lacking a formal security  analysis)~\cite{KinNad12,BLMR14,DCFZ18}, or make assumptions that break fungibility (e.g. honest majority in stake~\cite{FDKT+20,DCFZ18,karakostas2021securing} and/or static mining power~\cite{FDKT+20}).  
Indeed, with sufficient assumptions the problem of a hybrid protocol is trivially resolved. For instance, if we assume an honest majority in stake (at all  time), we can use a committee chosen randomly from the pool of stakeholders to assist a PoW ledger by regularly issuing {\em checkpoints}~\cite{karakostas2021securing}. However, the security of this scheme is solely guaranteed by the stakeholders, and the trust is entirely supported by PoS (not from  PoW).

\noindent {\bf Natural solutions}. 
In fact, if we assume a static setting (where both the total mining power and total `active' stake are fixed and known to the protocol designer), a simple solution is the following: PoW and PoS mining occur in parallel; and whichever PoW or PoS
succeeds first, goes ahead and extends the longest chain. However, there does not appear to be any straightforward (or otherwise) approach to extend this simple idea to the non-static setting, since we can no longer normalize the stake and work in the system when the total mining power is changing and unknown (\S\ref{sec:base}).  Even a fungible combination of work (i.e., total amount of honest work is more than total Byzantine work) emanating from two different hash functions (thus allowing  two  PoW blockchains using the same longest-chain consensus protocol to co-exist) has been unsolved.

\noindent {\bf Minotaur Protocol.} In this paper we present \minotaurnosp, a block\-chain protocol with proof of fungible work and stake.
At its core, \minotaur is a longest chain protocol that bases its block-proposer
schedule on our concept of \emph{virtual stake}  that fuses active {\em actual stake}
and {\em work stake}. 
%
The work-stake distribution is determined per epoch by measuring the participants' contributions in hashing power during a prior epoch. Work contribution is achieved by concurrent PoW mining of \emph{PoW blocks} (similar to `endorser inputs'~\cite{kiayias2017ouroboros} or `fruits'~\cite{pass2017fruitchains}) to be eventually referenced by main-chain blocks. The work-stake distribution is then assigned proportionally to a participant's share of endorser blocks referenced from the main chain during that epoch; and applying the fairness mechanisms in~\cite{GKL2015,pass2017fruitchains,kiayias2017ouroboros}, this assignment indeed truly represents the PoW-resource distribution among the participants.
\revadd{Note that this process can be seen as a fine-grained adaptation of PoW-based committee election such as in~\cite{PasShi16}.}
%
Figure~\ref{fig:minotaur} illustrates the protocol (a detailed description is given in \S\ref{sec:protocol}).

\noindent {\bf Minotaur security is optimal.} We show that \minotaur is  truly fungible between work and stake by showing it is secure when the sum of $\omega$ times the proportion of adversarial hash power ($\beta_w$) and $1-\omega$ times the proportion of the adversarial stake ($\beta_s$) is smaller than $1/2$, for any  $\omega \in [0,1]$, a weighing parameter of PoW/PoS (see Figure~\ref{fig:pow_pos} for $\omega = 0.5$). Figure~\ref{fig:pow_pos} also compares the security guarantee of \minotaur and a couple of non fungible  PoW/PoS protocols in the literature (the checkpointed ledger~\cite{karakostas2021securing}, the 2-hop blockchain~\cite{FDKT+20}, and a few finality gadgets~\cite{stewart2020grandpa,sankagiri2020blockchain,buterin2017casper,neu2021ebb}), with more details in Appendix~\ref{app:defend}.  We give a formal security analysis in \S\ref{sec:proof}. 
The new challenge that we have to tackle in our security analysis is to show a {\it fairness} guarantee in the work-stake conversion (e.g., a miner with 10\% mining power should gain 10\% work stake) in the \textit{non-static} setting and in the presence of an \textit{adversarial majority} of mining power. This requires a new understanding of how to bound the evolution of mining difficulty in the system compared to techniques used in previous works and the \bitcoin blockchain, and presents a significant barrier to surmount in our analysis.
One immediate consequence is that \minotaur can tolerate a 51\% mining adversary, as long as there is a supermajority in honest stake (and vice verse, a 51\% stake adversary). We prove our  security guarantee to be optimal, see \S\ref{sec:impos}, in the sense that otherwise the adversary will control the majority of the combined ``resource'' in the system and an analogy of Nakamoto's private chain attack~\cite{nakamoto2008bitcoin} could break the security (formal and detailed proof in \S\ref{sec:impos}). 
We also show that \minotaur is capable of tolerating fluctuations  of work much more effectively than \bitcoinnosp, cf. \S\ref{sec:fluctuation analysis}.

\begin{figure}
    \centering
    \begin{subfigure}{0.45\textwidth}
        \centering
        \includegraphics[width=\textwidth]{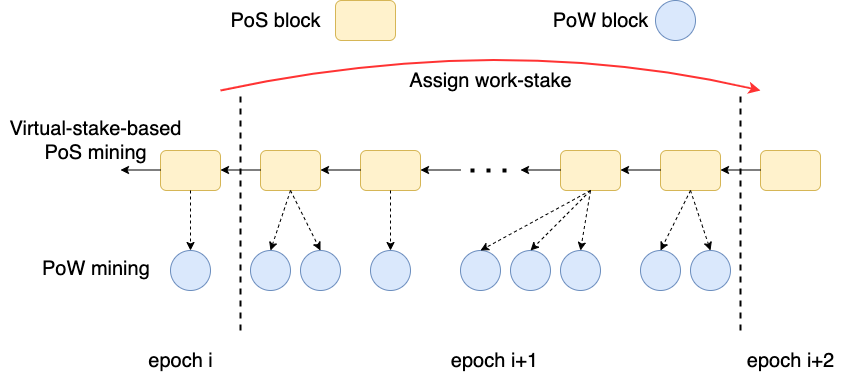}
        \caption{}
        \label{fig:minotaur}
    \end{subfigure}
    \hfill
    \begin{subfigure}{0.4\textwidth}
        \centering
        \includegraphics[width=\textwidth]{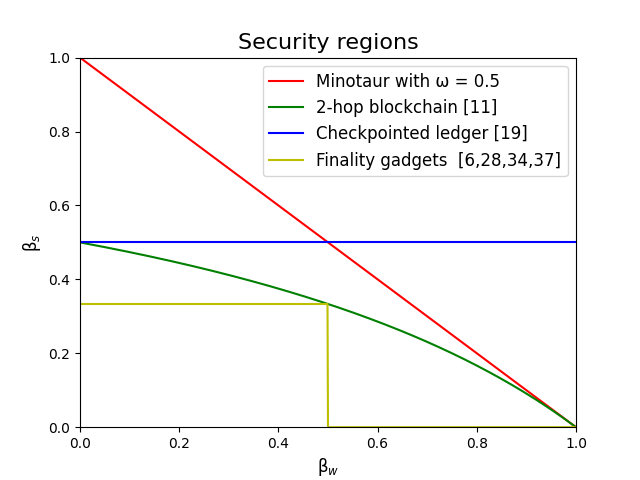}
        \caption{}
        \label{fig:pow_pos}
    \end{subfigure}
    \caption{(a) Architecture of \minotaur consensus protocol. (b) Comparison of security regions achieved by different hybrid PoW/PoS protocols (more details in Appendix~\ref{app:defend}).}
    \vspace{-0.5cm}
\end{figure}

\noindent {\bf Implementation.} We have implemented a prototype of a \minotaur client in about 6000 lines of Rust code~\cite{Minotaurcode} and provide experimental results to evaluate the \revdelete{concrete} performance of \minotaur under different scenarios. \revreplace{We also implemented a \bitcoin client as a benchmark.}{We also implemented \bitcoinnosp/\fruitchainsnosp/Ouroboros clients as benchmarks.} We use that client to validate also experimentally that \minotaur can survive more drastic variations of network hash power, while \bitcoinnosp/\fruitchains is no longer live under the same scenario. 
We also evaluate practical security concerns, like private chain, selfish mining and spamming attacks, on our \minotaur client. Our implementation shows that \minotaur performs well even under these attacks, and makes a stronger case for the practical viability and robustness of the system. \revreplace{More details on}{Details about} the implementation and the experiments \revreplace{can be found}{are given} in \S\ref{sec:exp}.

\revreplace[Request~\ref{req:motivation}\label{ans:motivation:2}]{ 
\noindent {\bf Launching considerations.} We point out that \minotaur provides also a secure and efficient method to bootstrap a PoW-based block\-chains. In general, it is more challenging to securely launch PoW blockchains, as a new blockchain would presumably start off with a relatively smaller total computation power and be vulnerable to the so-called ``51\% attacks''. On the other hand, PoS blockchains are easy to launch with existing techniques, such as proof-of-burn~\cite{karantias2020proof}, initial coin offering~\cite{li2018initial} and airdrop~\cite{airdrop}. The key advantage of  \minotaur is that it can launch in pure PoS mode (setting parameter $\omega=0$) and later transit into a pure PoW blockchain by gradually changing the weighing parameter $\omega$ to a higher value (\S\ref{sec:discuss}).
}
{ 
\noindent {\bf Motivation.}
Combining multiple resources helps to protect against scenarios where the adversary may gain control over the majority of one single resource but not over the combination of all involved resources.
For instance, to violate security in the PoW/PoS case with $\omega=0.5$, an adversary controlling half of the hash power would now additionally need to control half of the stake.

As a concrete example, consider the launch of a new blockchain. A PoS chain can be easily bootstrapped with existing techniques such proof-of-burn~\cite{karantias2020proof}, initial coin offering~\cite{li2018initial} and airdrop~\cite{airdrop}.
In contrast, the bootstrapping of a PoW chain is challenging, as the new system would presumably start off with a relatively small total computation power and thus be vulnerable to ``$51\%$ attacks''. Especially in view of the many PoW ledgers currently competing for the dedication of hashing power, a large established \bitcoin miner suddenly diverting their resources to the new system may easily result in their obtaining full control of the system.
To defend against such situations, \minotaur allows to launch a (projected) PoW blockchain in pure PoS mode (setting parameter $\omega=0$) and later transit into a pure PoW blockchain by gradually changing the weighing parameter $\omega$ to a higher value (\S\ref{sec:discuss}), thus allowing to safely ramp up the applied hashing power until a safe level of participation has been reached.
} 

\noindent {\bf Related work.}
The idea of hybrid PoW/PoS block production was first mentioned in~\cite{KinNad12}. A first concrete construction was given in~\cite{BLMR14} under the label `proof of activity', but without giving a rigorous security analysis. Their block-production mechanism essentially extends standard \bitcoin mining by having the mined block implicitly elect a set of stakeholders that are required to sign the block in order to validate it.

A similar construction was proposed in~\cite{FDKT+20} and proven secure under a majority of adversarial hashing power---however, their security proof still implies an honest majority of stake. In particular, and contrary to their initial claims, the protocol is not proven secure under the condition that \emph{any} minority of the combined resources is controlled by the adversary.

The work of~\cite{FDKT+20} was extended in~\cite{DCFZ18}, to adapt to dynamic participation of both, hashing power and stake, and in~\cite{TNDF+18}, to combine PoW with multiple resources (rather than just PoS)---both works lack a rigorous security analysis.

The application of  PoW/PoS hybrid block production for the goal of protecting early-stage PoW systems against initial periods of adversarial dominance in hashing power was explored in \cite{CXGL+18}. They propose to start the system with hybrid block production (where each resource contributes to a certain fraction of the blocks) and then slowly fade out the stake contribution to eventually turn the system into pure PoW. Also this protocol is not proven secure.

There also exist another class of  hybrid PoW/PoS protocols~\cite{neu2021ebb,sankagiri2020blockchain,buterin2017casper,stewart2020grandpa}, which uses a BFT protocol (PoS) to build a finality gadget/layer on the top of a Nakamoto-style longest chain protocol (PoW) to achieve important properties such as accountability and finality. However, these protocols require an honest majority (or even supermajority) on both the set of miners and the set of stakeholders hence they are not fungible. 


\ignore{
The 2-hop blockchain~\cite{FDKT+20} is also a hybrid PoW/PoS protocol, which aims to survive the 51\% mining attack. In a 2-hop blockchain, there are two coupled blockchains, one  based a PoW mechanism and another using a PoS-based approach, and these two blockchains are extended alternately. However, the protocol is only designed for the static computing power and stake setting, which makes it unpractical. In addition, \minotaur is more flexible as weighing PoW and PoS can be supported.  
}

\section{Preliminaries}\label{sec:prelim}

\subsection{Security model}

\noindent {\bf Time, slots, and synchrony.} We consider a setting where time is divided into discrete units called \textit{slots}. Players are equipped with (roughly synchronized) clocks that indicate the current slot. Slots are indexed by integers, and further grouped into epochs with fixed size $R$, i.e., epoch $e$ composes of slots $\{(e-1)R+1,(e-1)R+2,\cdots,eR\}$. And we assume that the real time window that corresponds to each slot has the following two properties: (1) The current slot is determined by a publicly-known and monotonically increasing function of current time; (2) Each player has access to the current time and any discrepancies between parties’ local time are insignificant in comparison with the duration of a slot.

We describe our protocols in the now-standard  $\Delta$-synchronous network model considered in~\cite{pass2017analysis,garay2020full,david2018ouroboros,badertscher2018ouroboros}, where there is an (unknown) upper bound $\Delta$ in the delay (measured in number of slots) that the adversary may inflict to the delivery of any message. 
Similar to~\cite{garay2020full,pass2017analysis}, the protocol execution proceeds in slots with inputs provided by an environment program denoted by $\Z(1^\kappa)$ to parties that execute the protocol $\Pi$, where $\kappa$ is a security parameter. 
The network is modeled as a diffusion functionality similar to those in~\cite{garay2020full,pass2017analysis}: it allows message ordering to be controlled by the adversary $\A$, i.e., $\A$ can inject messages for selective delivery but cannot change the contents of the honest parties’ messages nor prevent them from being delivered within $\Delta$ slots of delay — a functionality parameter. Specially, for $\Delta = 1$, the network model is reduced to the so-called {\it lock-step synchronous model}, where messages are guaranteed to be delivered within one slot.

\noindent {\bf Protocol player.} The protocol is executed by two types of players, PoW-miners (miners for short) and PoS-holders (stakeholders for short), who generate different types of blocks, PoW blocks and PoS blocks. Note that we allow for any relation among the sets of miners and stakeholders, including the possibility that all players play both roles, or the two types of players are disjoint.

\noindent{\bf Random oracle.} For PoW mining, we abstract the hash function as a random-oracle functionality. It accepts queries of the form $(\texttt{compute}, x)$ and $(\texttt{verify}, x, y)$. For the first type of query, assuming $x$ was never queried before, a value $y$ is sampled from $\{0, 1\}^\kappa$ and it is entered to a table $T_H$. If $x$ was queried before, the pair $(x, y)$ is recovered from $T_H$. In both cases, the value $y$ is provided as an answer to the query. For the second type of query, a lookup operation is performed on the table. Honest miners are allowed to ask one query per slot of the type \texttt{compute} and unlimited queries of the type \texttt{verify}. The adversary $\A$ is given a bounded number of \texttt{compute} queries per slot and also unlimited number of \texttt{verify} queries. The bound for the adversary is determined as follows. Whenever a corrupted miner is activated the bound is increased by 1; whenever a query is asked the bound is decreased by 1 (it does not matter which specific corrupted miner makes the query).

\noindent{\bf Adversarial control of resources.} We assume a Byzantine adversary who can decide on the spot how many honest/corrupted miners 
are activated adaptively. 
\revadd[Request~\ref{req:adaptivity}\label{ans:adaptivity}]{Note that we allow instantaneously adaptive corruption on miners, which means that both the number of honest miners and the number of corrupted miners in each slot are chosen by the adversary. For stakers, the adversary is only ‘moderately’ adaptive (participant corruption only takes effect after a certain delay), just like in Ouroboros classic~\cite{kiayias2017ouroboros}.}
In a slot $r$, the number of honest miners that are active in the protocol is denoted by $h_r$ and the number of all active miners in slot $r$ is denoted by $n_r$. Note that $h_r$ can only be determined by examining the view of all honest miners and is not a quantity that is accessible to any of the honest miners individually. In order to obtain meaningful concentration bounds on the number of PoW blocks in a long enough window, we set a lower bound $\alpha_0$ on the fraction of honest mining power, i.e., $h_r \geq \alpha_0 n_r$ for all $r$. Note that $\alpha_0$ does not have to be $1/2$ as required by \bitcoinnosp, it can be a small positive constant.
Sudden decreases of total mining power could hurt the liveness of the protocol due to too few PoW blocks mined in one epoch. Therefore, we need to restrict the fluctuation of the number of honest/adversarial queries in a certain limited fashion. Suppose $\Z, \A$ with fixed coins 
produces a sequence of honest miners $h_r$, where $r$ ranges over all slots of the entire execution, we define the following notation.

\begin{defn}[from~\cite{garay2017bitcoin}]
\label{def:respecting}
For $\gamma \in [1,\infty)$, we call a sequence $(x_r), r \in [LR]$, as $(\gamma, s)$-respecting if for any set $S \subseteq [0, LR]$ of at
most $s$ consecutive slots, 
$$\max_{r \in S}\  x_r \leq \gamma \min_{r\in S} \  x_r,$$
We say that $\Z$ is $(\gamma, s)$-respecting if for all $\A$ and coins for $\Z$ and $\A$, both the sequences of $(h_r)$ and $(n_r)$ are $(\gamma, s)$-respecting.
\end{defn}

Finally, \minotaur achieves consensus via proof of fungible work and stake. The following is the formal definition of fungibility and our major assumptions on the adversarial power.

\begin{defn}[Fungibility of resources]
\label{def:adv}
For a time window $W$, let $\beta_s^{W}$ be the maximum fraction of adversarial stake in $W$, where the maximum is taken over all views of all honest players across all slots in $W$; and let $\beta_w^{W}$ be the maximum fraction of adversarial mining power over all slots in $W$ (i.e., $\beta_w^{W} = 1 - \min_{r \in W}h_r/n_r$). For $\theta \in [0,1]$, we say the adversary $\A$ is $(\theta,m,\omega)$-bounded if for any time window $W$ with at most $m$ slots, we have $\omega\beta_w^W + (1-\omega)\beta_s^W \leq \theta$. We say a blockchain protocol achieves fungibility of work and stake if it is secure against such an adversary.
\end{defn}

To keep the paper simple, the \minotaur protocol, as described, will rely on  Assumption~\ref{ass:main1} below
that initializes the protocol with honest majority of stake;
however this assumption is not essential; we refer the reader to \S\ref{sec:discuss} for a discussion on how to remove it. 

\begin{assu}[Initialization]
\label{ass:main1}
During the initialization phase of the protocol, we assume:
\begin{enumerate}
    \item[1.1] The initial stake distribution has honest majority, i.e., $\beta_s^0 \leq 1/2-\sigma$, where $\beta_s^0$ is the fraction of initial stake controlled by the adversary.
    \item[1.2] We have a good estimate of the initial honest mining power $h_1$. In particular, let $\tilde h_1$ be the estimate of $h_1$, we have $h_1/(1+\delta)\gamma^2 \leq \tilde h_1 \leq \gamma^2 h_1 /(1-\delta)\alpha_0$.
\end{enumerate}
\end{assu}

\revadd{
Post initialization, we relax the above to our intended fungible resource assumption described below. }

\begin{assu}[Execution]
\label{ass:main2}
During the execution phase of the protocol, we assume:
\begin{enumerate}
    \item[2.1] Stake-work bound: The adversary $\A$ is $(1/2-2\sigma,2R,\omega)$-bounded.
    \item[2.2] Work fluctuation bound: The environment $\Z$ is $(\gamma,2R)$-respecting.
    \item[2.3] Work participation bound: For any slot $r \in [LR]$, $h_r \geq \alpha_0 n_r$.
\end{enumerate}
\end{assu}

\noindent {\bf Participation model.} \minotaur can be constructed based on different PoS longest chain protocols. Various versions of \minotaur will take different subsets of the following assumptions on the participation model:
\begin{enumerate}
    \item[(P1)] Honest stakeholders are always online\footnote{This assumption can be easily relaxed. Namely, it is sufficient for an honest stakeholder to come online at the beginning of each epoch, determine whether it belongs to the slot leader set for any slots within this epoch, and then come online and act on those slots while maintaining online presence at least every $k$ slots. See Appendix H of~\cite{david2018ouroboros} for more detail.};
    \item[(P2)] Honest miners who mined PoW blocks in epoch $e$ will stay online in epoch $e+2$;\footnote{This can be relaxed similarly as P1.}
    \item[(P3)] In case an honest party joins after the beginning of the protocol, its initialization chain $\mathcal{C}$ provided by the environment should match an honest party’s chain which was active in the previous slot.
\end{enumerate}

\subsection{Blockchain security properties}
\label{sec:prop}

\begin{nota}
\label{notation:chains}
We denote by $\mathcal{C}^{\lceil \ell}$ the chain resulting from ``pruning'' the blocks with timestamps within the last $\ell$ slots. If $\mathcal{C}_1$ is a prefix of $\mathcal{C}_2$, we write $\mathcal{C}_1 \prec \mathcal{C}_2$. The latest block in the chain $\mathcal{C}$ is called the head of the chain and is denoted by head$(\mathcal{C})$. We denote by $\mathcal{C}_1 \cap \mathcal{C}_2$ the common prefix of chains $\mathcal{C}_1$ and $\mathcal{C}_2$. Given a chain $\mathcal{C}$ and an interval $S$ (or $[r_1,r_2]$), let $\mathcal{C}(S)$ (or $\mathcal{C}[r_1,r_2]$) be the segment of $\mathcal{C}$ containing blocks with timestamps in $S$ (or $[r_1,r_2]$). We say that a chain $\mathcal{C}$ is held by or belongs to an honest node if it is one of the best chains in its view. 
\end{nota}

\begin{defn}[Common Prefix (CP)]
\label{def:common}
The common prefix property with parameter $\ell_{\rm cp} \in \mathbb{N}$ states that for any two honest nodes holding chains $\mathcal{C}_1$, $\mathcal{C}_2$ at slots $r_1$, $r_2$, with $r_1 \leq r_2$, it holds that $\mathcal{C}_1^{\lceil \ell_{\rm cp}} \prec \mathcal{C}_2$.
\end{defn}

\begin{defn}[Existential Chain Quality ($\exists$CQ)]
\label{def:quality}
The existential chain quality property with parameter $\ell_{\rm cq} \in \mathbb{N}$ states that for any chain $\mathcal{C}$ held by any honest party at slot $r$ and any interval $S \subseteq [0,r]$ with at least $\ell_{\rm cq}$ consecutive slots, there is at least one honestly generated block in $\mathcal{C}(S)$.
\end{defn}

Our goal is to generate a robust transaction ledger that satisfies {\em persistence}  and {\em liveness} as defined in~\cite{kiayias2017ouroboros}.

\begin{defn}
    \label{def:public_ledger}
    A protocol $\Pi$ maintains a robust public transaction ledger if it organizes the ledger as a blockchain of transactions and it satisfies the following two properties:
    \begin{itemize}
        \item (Persistence) Consider the confirmed ledger $L_1$ on any node $p_1$ at any slot $r_1$, and the confirmed ledger $L_2$ on any node $p_2$ at any slot $r_2$ (here $u_1$ ($r_1$) may or may not equal $u_2$ ($r_2$)).  If $r_1 +\Delta < r_2$, then $L_1$ is a prefix of $L_2$.
        \item (Liveness) Parameterized by $u \in \mathbb{R}$, if a transaction {\sf tx} is received by all honest nodes for more than $u$ slots, then all honest nodes will contain {\sf tx} in the same place in the confirmed ledger.
    \end{itemize}
\end{defn}
\section{Impossibility result}
\label{sec:impos}

Consider any protocol $\Pi$ executed by two types of players, miners and stakeholders, and let $D_\Pi$ be the region of $(\beta_w, \beta_s)$ such that $\Pi$ can generate a robust public transaction ledger under an adversary controlling a $\beta_w$ fraction of mining power and a $\beta_s$ fraction of stake (see Figure~\ref{fig:pow_pos}).

\begin{theorem}
\label{thm:converse}
Any two points, $X_1 = (p_1, q_1)$ and $X_2 = (p_2, q_2)$ such that $p_1 + p_2 \geq 1$ and $q_1 + q_2 \geq 1$, cannot co-exist in $D_\Pi$. 
\end{theorem}
\begin{proof}
We show that if there exists a protocol $\Pi$ secure under both points $X_1$ and $X_2$, then we can use $\Pi$ to implement a robust transaction ledger with two players, one of them being Byzantine, which is impossible.

Let Alice and Bob be two players, one of them possibly malicious. Let Alice control a $p_1$ fraction of mining power and a $1-q_2$ fraction of stake. Let Bob control a $1-p_1$ fraction of mining power and a $q_2$ fraction of stake.
In the case where Alice is malicious, a $p_1$ fraction of mining power and a $1-q_2 \leq q_1$ fraction of stake are malicious, which is dominated by point $X_1$. In the case where Bob is malicious, a $1-p_1 \leq p_2$ fraction of mining power and a $q_2$ fraction of stake are malicious, which is dominated by point $X_2$. Assuming that $\Pi$ implements a robust public transaction ledger implies that also the protocol among the two players does while one of them is malicious. However, by Theorem~1 in~\cite{garay2020sok}, tolerating $f=1$ malicious player requires at least $2f+1=3$ players. This is a contradiction.
\end{proof}

Figure~\ref{fig:possible_regions} gives a few examples of possible {\it maximum} security regions that satisfy the constraint by Theorem~\ref{thm:converse}. By maximum, we mean the region cannot be enlarged. All these regions are bounded by a \textit{non-increasing} curve that is \textit{symmetrical} about the point $(1/2,1/2)$. Thus, all these security regions have area $1/2$, which is an analogy of the $1/2$ fault tolerance in the single-resource systems. In this work, we focus on achieving the security regions bounded by a linear curve (the yellow and red ones in Figure~\ref{fig:possible_regions}) and leave the achievability of other possible regions (e.g., the green and blue ones) to future work. \revadd{Note that the security region achieved by the checkpointed ledger~\cite{karakostas2021securing} is also optimal (albeit not fungible), and we will see that this region can also be achieved by \minotaur with $\omega= 0$.}

\begin{figure}
    \centering
    \includegraphics[width=0.4\textwidth]{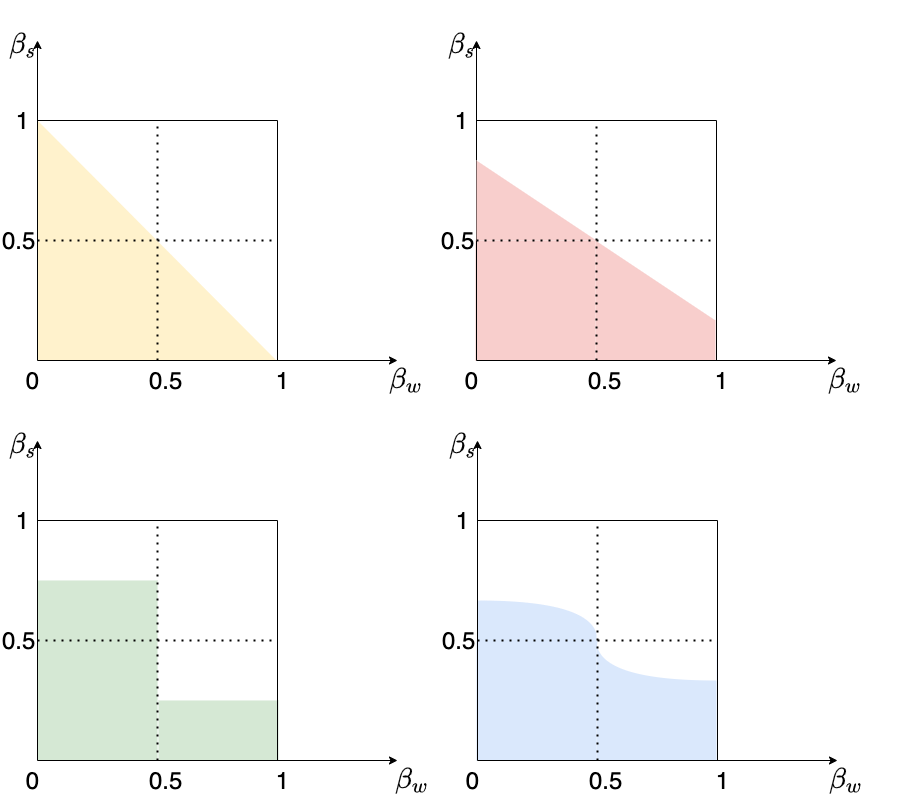}
    \caption{A few examples of possible security regions satisfying the constraint by Theorem~\ref{thm:converse}.}
    \label{fig:possible_regions}
    \vspace{-0.5cm}
\end{figure}

\begin{cor}
\label{cor:adv}
No protocol $\Pi$ can generate a robust public transaction ledger under a $(1/2,m,\cdot)$-bounded adversary for any $m$.
\end{cor}

Corollary~\ref{cor:adv} thus implies that Assumption~\ref{ass:main2}.1 is not only sufficient but also necessary.

\section{Baseline approaches}
\label{sec:base}

In this section, we provide a couple of natural designs of hybrid PoW/PoS protocols, which fail to achieve the full goal of \minotaurnosp.

\noindent {\bf Idea 1: Securing PoW chain via checkpointing.} The checkpointed ledger~\cite{karakostas2021securing} employs an external set of parties or a committee chosen randomly from the pool of stakeholders to assist a PoW ledger by finalizing blocks shortly after their creation. The finalized blocks are called checkpoints and the final ledger is formed by the chain of checkpoints. This mechanism can also secure a PoW ledger in the presence of adversarial mining majorities, but the security is solely guaranteed by the external set or the stakeholders.

Other checkpointing protocols~\cite{neu2021ebb,sankagiri2020blockchain,buterin2017casper,stewart2020grandpa} achieve properties that are not possible with pure PoW protocols, such as accountability and finality.
\revreplace[Request~\ref{req:51}\label{ans:51}]{ 
However, these protocols require an honest majority (or even supermajority) on \emph{both} the set of miners and the set of stakeholders.
}
{
However, in contrast to our requirement to tolerate any adversarial minority of the \emph{combined} resources, these protocols are based on the stronger assumption of honest majority (or even supermajority) of \emph{both} the set of miners and the set of stakeholders.
} 

\noindent {\bf Idea 2:  Smooth interpolation among PoW and PoS blocks.} In the static setting (where both the total mining power and total active stake are fixed and known to the protocol designer), there is a simple protocol that can also achieve the regions defined by the red line in Figure~\ref{fig:pow_pos}. In this protocol, PoW and PoS mining occur in parallel, following the longest chain rule. In the initialization phase, we tune the mining targets such that the mining rate (i.e., number of blocks produced per unit time) of PoW blocks is $\omega f$ and the mining rate of PoS blocks is $(1-\omega)f$. In the execution phase, whichever miners or stakeholders succeeds first, it goes ahead and extends the longest chain \revadd{which accepts both types of blocks.}

Compared with a pure PoS protocol (e.g., Ouroboros Praos\revreplace{~\cite{badertscher2018ouroboros}}{~\cite{david2018ouroboros}}), the adversary \revdelete{in this protocol }has strictly smaller action space because it cannot equivocate with the PoW blocks. Thus, the security of this protocol follows directly from the security of Ouroboros Praos via either the forkable string argument~\cite{badertscher2018ouroboros} or the Nakamoto block method developed in~\cite{dembo2020everything}. We remark that the existence of this simple protocol makes the 2-hop blockchain~\cite{FDKT+20} less interesting as it only works in the static setting and does not even allow weighing PoW and PoS.

\revhighlight[Request~\ref{req:51}\label{ans:51:2}]{
However, it is very hard to extend this simple idea to the non-static setting, particularly with variable mining power. A natural approach to support variable mining power is to have fixed-length epochs and adjust the mining target of PoW blocks every epoch as in \bitcoinnosp. From the simple protocol described above, we learn that in order to guarantee security in the non-static setting, we need to make sure that the ratio of total mining rates of PoW blocks and PoS blocks is a constant $\omega/(1-\omega)$ in every epoch. However, this is impossible to achieve in the variable mining power setting because the adversary can always decide to hide or release its PoW blocks so that there is no way to estimate the total mining power accurately. Inaccurate PoW mining target adjustment could lead to a different weighing parameter $\omega$ of the security region than the desired one. More importantly, the value of $\omega$ can be easily manipulated by the adversary (via selectively publishing its PoW blocks), and is unknown to the honest players.
}
\section{The full protocol}
\label{sec:protocol}



In this section, we provide a detailed description of our protocol \minotaurnosp. The protocol is built on an epoch-based PoS longest chain protocol \revreplace{(e.g., Ouroboros~\cite{kiayias2017ouroboros}, Ouroboros Praos~\cite{david2018ouroboros}, Ouroboros Genesis~\cite{badertscher2018ouroboros})}{e.g., Ouroboros Classic~\cite{kiayias2017ouroboros}, Praos~\cite{david2018ouroboros}, or Genesis~\cite{badertscher2018ouroboros}}. Recall that in an epoch-based PoS protocol, time is divided into multiple epochs, each with a fixed number of slots. In each slot, one or multiple stakeholders are selected as block proposers by a PoS `lottery' process. In this process, the probability of being selected is proportional to the relative stake a stakeholder has in the system, as reported by the blockchain itself. \revadd{See Appendix~\ref{app:ouroboros} for more details.}

In contrast to the original Ouroboros protocols, the block-pro\-poser schedule accounts for both resources by means of \emph{virtual stake}, which is a combination of the \emph{actual stake} (representing stake as in the original protocol), and \emph{work stake} representing the share in block-production rights attributed to the PoW resource.

Additionally, PoW miners participate by mining \emph{endorser blocks} (along the lines of `endorser inputs'~\cite{kiayias2017ouroboros} or `fruits'~\cite{pass2017fruitchains}). These endorser blocks are not directly appended to the main chain, but are to be referenced by future main-chain blocks (i.e., PoS blocks scheduled by means of virtual stake). Each epoch is assigned a certain amount of work stake that is assigned to the PoW miners who succeed in mining PoW blocks;
the work stake assigned to the respective miners is proportional to their contribution of PoW blocks referenced from the main chain during an epoch.
PoS block production rights per epoch are then assigned by considering
the sum of actual stake (contributed by tokens) and work stake
(contributed by PoW blocks).

Note that the purpose of the PoW endorser blocks is to measure work fairly, not implying that they also need to carry the ledger transactions. Transaction inclusion is orthogonal to this aspect, and can still be implemented along the lines of \bitcoin (transaction inclusion in the main-chain blocks) or Input-Endorsers/\fruitchains (transaction inclusion in the endorser blocks), or variants thereof.

We now treat the underlying PoS protocol as a black-box and present the protocol in detail. The protocol runs in fixed-time epochs with $R$ slots each.
\begin{itemize}
\item PoW mining: In slot $sl$ of epoch $e$, a miner mines a PoW block if
$$H(pk,sl,h,mr,nonce) < T_e,$$
where $pk$ is the public key of the miner, $h$ is the hash of the last confirmed PoS block (according to the confirmation rule of the underlying PoS protocol), $mr$ is the Merkle root of the payload, $T_e$ is the mining target in epoch $e$. Like in PoW blockchains, miners try different values of $nonce$ to solve this hash puzzle. Following the notation in~\cite{garay2020full}, we define the \textit{difficulty} of this PoW block to be $1/T_e$. Besides possible transactions, also the public key of the miner is included in the payload. A PoW block is called recent in current slot $sl_{\rm now}$ if it refers to a confirmed PoS block mined no earlier than slot $sl_{\rm now}-sl_{\rm re}$, where $sl_{\rm re}$ is called the \emph{recency parameter}. 
\item PoS mining: in slot $sl$ of epoch $e$, one or multiple stakeholders are selected to propose PoS blocks extending the best chain (according to the chain selection rule). The selection of PoS block proposers uses the same mechanism as in the underlying PoS protocol. But for each node, the probability of being selected is proportional to its relative ``virtual'' stake, instead of the relative actual stake. The virtual stake is the sum of actual stake and work stake. At the beginning of each epoch, we set the total work stake equal to $\frac{\omega}{1-\omega}$ times the total actual stake in the system. And the work stake distribution used in epoch $e$ is set to be the distribution of block difficulties from PoW blocks referred by PoS blocks in epoch $e-2$.
\item Adjust the PoW mining target: For epoch $e$, $T_e$ is adjusted according to $D_{e-1}^{\rm total}$ the total difficulty of PoW blocks referred by PoS blocks in slots $[(e-2)*R-k + 1, (e-1)*R-k]$, i.e., $T_e = f^{(w)}R/D_{e-1}^{\rm total}$, where $f^{(w)}$ is a protocol parameter representing the expected PoW mining rate in number of blocks per slot. 
(Specially, for $e=2$, $T_2 = f^{(w)}(R-k)/D_{1}^{\rm total}$.)
Note that around the boundary of two epochs $e-1$ and epoch $e$, the adversary has the option to use targets from both epochs, as PoW blocks don’t have accurate timestamps.
\item Chain selection rule: due to a long-range attack (see Appendix~\ref{app:long_range}), the longest chain rule would fail. We use the following chain selection rules, which have different security guarantees.
\begin{itemize}
    \item {\sf maxvalid-mc} 
    (from Ouroboros Praos~\cite{david2018ouroboros}, needs a trusted third party for bootstrapping, i.e., assumption (P3)): it prefers longer chains, unless the new chain forks more than $k$ blocks relative to the currently held chain (in which case the new chain would be discarded).
    \item {\sf maxvalid-bg} (from Ouroboros Genesis~\cite{badertscher2018ouroboros}, supports bootstrapping from genesis block): it prefers longer chains, if the new chain $\mathcal{C}'$ forks less than $k$ blocks relative to the currently held chain $\mathcal{C}$. If $\mathcal{C}'$ did fork more than $k$ blocks relative to $\mathcal{C}$, $\mathcal{C}'$ would still be preferred if it grows more quickly in the $s$ slots following the slot associated with the last common block of $\mathcal{C}$ and $\mathcal{C}'$.
\end{itemize}

\end{itemize}
\section{Security Analysis}
\label{sec:proof}

\subsection{Main theorems}

Under different security models and assumptions, we prove the security (in particular, {\it persistence} and {\it liveness} as defined in \S\ref{sec:prop}) of \minotaur constructed with three various versions of Ouroboros PoS protocols.

\begin{theorem}
\label{thm:main1}
Under Assumption~\ref{ass:main1}, Assumption~\ref{ass:main2} and assumptions (P1)-(P3), when executed in a lock-step synchronous model (i.e., $\Delta$-synchronous model with $\Delta = 1$), \minotaur constructed with Ouroboros~\cite{kiayias2017ouroboros} generates a transaction ledger that satisfies {\it persistence} and {\it liveness} with overwhelming probability.
\end{theorem}

\begin{theorem}
Under Assumption~\ref{ass:main1}, Assumption~\ref{ass:main2} and assumptions (P1)-(P3), when executed in a $\Delta$-synchronous model, \minotaur constructed with Ouroboros Praos~\cite{david2018ouroboros} generates a transaction ledger that satisfies {\it persistence} and {\it liveness} with overwhelming probability.
\end{theorem}

\begin{theorem}
Under Assumption~\ref{ass:main1}, Assumption~\ref{ass:main2} and assumptions (P1)-(P2), when executed in a $\Delta$-synchronous model, \minotaur constructed with Ouroboros Genesis~\cite{badertscher2018ouroboros} generates a transaction ledger that satisfies {\it persistence} and {\it liveness} with overwhelming probability.
\end{theorem}

\subsection{Security with Ouroboros}
\label{sec:mc}


We first provide a proof sketch for Theorem~\ref{thm:main1}. All the parameters used in our analysis are listed in Table~\ref{tbl:notation}.
\begin{table}[h]
{\small
\centering
\begin{tabular}{ |c l| }
\hline
    $h_r$ & number of honest PoW queries in slot $r$ \\
    $n_r$ & number of total PoW queries in slot $r$ \\
    $\alpha_0$ & lower bound on the fraction of honest mining power \\
    $sl_{\rm re}$ & recency parameter for PoW blocks \\
    $\Delta$ & network delay in slots \\
    $\kappa$ & security parameter; length of the hash function output  \\
    $R \in \mathbb{N}$ &  duration of an epoch in number of slots\\
    $(\gamma, s)$ & restriction on the fluctuation of the number of \\
    & honest queries across slots (Definition~\ref{def:respecting}) \\
    $f^{(s)}$ & expected PoS mining rate in number of blocks per slot \\
    $f^{(w)}$ & expected PoW mining rate in number of blocks per slot \\
    $\epsilon$ & quality of concentration of random variables \\
    $\sigma$ & PoW fairness parameter \\
    & \& advantage of honest parties (Assumption~\ref{ass:main1}.1 and \ref{ass:main2}.1) \\
    $\delta$ & ``goodness'' parameter of an epoch/slot (Definition~\ref{def:good}) \\
    $\lambda$ & typicality parameter of the execution \\
    $\ell$ & minimum number of slots for concentration bounds\\
    $L$ & total number of epochs in the execution\\
    $\omega$ & the weighing parameter \\ 
 \hline
\end{tabular}
\caption{The parameters used in our analysis.}
\label{tbl:notation}
}
\end{table}

\begin{proof}[Proof sketch.]
The proof relies on two important arguments:
\begin{enumerate}
    \item[1)] {\bf Single-epoch argument.} Given honest majority in the virtual stake (normalized work-stake + normalized actual-stake) used in one epoch, we prove the security properties (CP and $\exists$CQ) for this single epoch. For this, we use the so-called {\it forkable strings} technique developed in the Ouroboros papers~\cite{badertscher2018ouroboros,david2018ouroboros,kiayias2017ouroboros}.
    \item[2)] {\bf FruitChain argument.} Given CP and $\exists$CQ of the PoS chain in an epoch, we prove {\it fairness} of PoW blocks in this epoch, i.e., a miner controlling a $\phi$ fraction of the computational resources will contribute a $\phi$ fraction of work. 
\end{enumerate}

Using these building blocks, we prove the security with {\sf maxvalid-mc} inductively. Under the honest majority assumption in the initial stake distribution (Assumption~\ref{ass:main1}.1), we can prove security of the PoS chain in epochs 1\&2. Applying the \fruitchains argument, we get fairness of PoW blocks in epochs 1\&2. Combining with Assumption~\ref{ass:main2}.1, i.e., $\omega\beta_w+(1-\omega)\beta_s < 1$ (for any long enough window), we have honest majority in the virtual stake distribution used in epoch 3\&4. And the proof goes on till the last epoch (Figure~\ref{fig:proof_sketch}).
\end{proof}

\begin{figure*}
    \centering
    \includegraphics[width=0.8\textwidth]{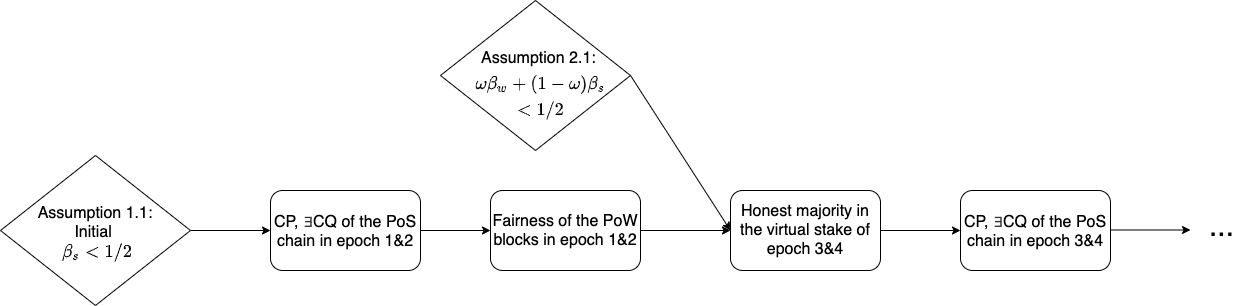}
    \vspace{-0.4cm}
    \caption{Proof sketch for \S\ref{sec:mc}.}
    \vspace{-0.4cm}
    \label{fig:proof_sketch}
\end{figure*}

\subsubsection{Single-epoch security}

\begin{theorem}[Theorem 3 from~\cite{david2018ouroboros}]
\label{thm:single}
Let $\kappa,R,\Delta \in \mathbb{N}$ and $\sigma \in (0,1)$. Let $\beta_v$ be the fraction of adversarial virtual stake satisfying 
$$\beta_v \leq 1/2 - \sigma$$
for some positive constant $\sigma$.
Then the probability that the adversary violates CP with parameter $\ell_{\rm cp} = \kappa$ and $\exists$CQ with parameter $\ell_{\rm cq} = \kappa$ throughout a period of $R$ slots is no more than $R e^{-\Omega(\kappa)}$. The constant hidden by the $\Omega(\cdot)$-notation depends on $\sigma$.
\end{theorem}

\revadd[Request~\ref{req:virtualstake}\label{ans:virtualstake}]{In \minotaurnosp, the main-chain blocks are scheduled as a function of the virtual-stake distribution only. The Ouroboros analysis thus directly translates by replacing its actual stake by virtual stake.}

\subsubsection{FruitChain argument}

\label{sec:fruits}

Recall that $h_r$ is the number of honest PoW queries in slot $r$ and $n_r$ is the number of total PoW queries in slot $r$. For a set of slots $S$, we define $h(S) = \sum_{r\in S}h_r$ and $n(S)= \sum_{r\in S} n_r$. In order to obtain meaningful concentration bounds on the number of PoW blocks in one epoch, there is a lower bound $\alpha_0$ on the fraction of honest mining power, i.e., $h_r \geq \alpha_0 n_r$ for all $r$. 

In the analysis of this subsection, we assume the main chain (PoS chain) satisfies properties CP with parameter $\ell_{\rm cp} = \kappa$ and $\exists$CQ with parameter $\ell_{\rm cq} = \kappa$.
By the common prefix property, for a PoS chain $\mathcal{C}$ held by an honest node at slot $r$, the prefix $\mathcal{C}^{\lceil \kappa}$ are stabilized, so to mine a PoW block at slot $r$ an honest miner will refer the tip of $\mathcal{C}^{\lceil \kappa}$ as the last confirmed PoS block. And we set the recency parameter $sl_{\rm re} = 3\kappa + \Delta$, i.e., a PoW block $B_w$ is recent w.r.t. a chain $\mathcal{C}$ at slot $r$ if the confirmed PoS block referred by $B_w$ is in $\mathcal{C}$ and has timestamp at least $r-3\kappa-\Delta$. With this selection of the recency parameter, we can prove the following key property of the protocol: any PoW block mined by an honest miner will be incorporated into the stabilized chain (and thus never lost). We refer to this as the {\it Fruit Freshness Lemma}—PoW blocks stay ``fresh'' sufficiently long to be incorporated. 

\begin{lemma}[Fruit Freshness]
\label{lem:fresh}
Suppose the PoS chain satisfies properties CP with parameter $\ell_{\rm cp} = \kappa$ and $\exists$CQ with parameter $\ell_{\rm cq} = \kappa$. Then, if $sl_{\rm re} = 3\kappa+\Delta$, an honest PoW block mined at slot $r$ will be included into the stabilized chain before slot $r+r_{\rm wait}$, where $r_{\rm wait} = 2\kappa+\Delta$.
\end{lemma}

\begin{proof}
Suppose an honest PoW block $B_w$ is mined at slot $r_0$ while the PoS chain $\mathcal{C}_0$ is adopted, then $B_w$ will reference the tip of $\mathcal{C}_0^{\lceil \kappa}$ as the last confirmed PoS block. Further, by the $\exists$CQ property, the tip of $\mathcal{C}_0^{\lceil \kappa}$ has timestamp $r_1 \geq r_0 - 2\kappa$. By slot $r_0 +\Delta$, all honest nodes will receive $B_w$. Let $r_2 = r_0 + 2\kappa +\Delta$ and $\mathcal{C}$ be any chain held by an honest node at slot $r_2$, then again by $\exists$CQ, there exists \revdelete{at least one}{an} honest block $B_s$ on $\mathcal{C}$ whose timestamp $r_3$ is in the interval $[r_2-2\kappa, r_2-\kappa)$. We check that $B_w$ is still recent at slot $r_3$ as $$r_3 - r_1 < (r_2 - \kappa) - (r_0 - 2\kappa) = 3\kappa + \Delta  = sl_{\rm re}.$$
As $B_s$ is an honest block mined after $r_0+\Delta$, $B_s$ or an ancestor of $B_s$ must include $B_w$. And since $B_s$ is stabilized in $\mathcal{C}$ at slot $r_2$, we have that $r_{\rm wait} = r_2 - r_0 = 2\kappa +\Delta$.
\end{proof}



In a $(\gamma,s)$-respecting environment, we have the following useful proposition.

\begin{prop}[Proposition 2 from \cite{garay2020full}]
\label{prop:envbounds}
In a $(\gamma, s)$-respecting environment, let $U$ be a set of at most $s$ consecutive slots and $S \subseteq U$. Then, for any $h \in \{h_r : r \in U\}$ and any $n \in \{n_r : r \in U\}$ we have
\begin{align*}
    \frac{h}{\gamma} &\leq \frac{h(S)}{|S|} \leq \gamma h,~~~~ \frac{n}{\gamma} \leq \frac{n(S)}{|S|} \leq \gamma n, \\
    h(U) &\leq\Bigl(1 + \frac{\gamma |U \setminus S|}{|S|}\Bigr)h(S),~~~~ n(U) \leq\Bigl(1 + \frac{\gamma |U \setminus S|}{|S|}\Bigr)n(S). 
\end{align*}
\end{prop}

Recall that $T_e$ is the mining target in epoch $e$ determined by the stabilized segment of the chain from epoch $e-1$. In order to obtain meaningful concentration bounds on the number of PoW blocks, we need $T_e$ to be `reasonable' for each epoch. Similar to~\cite{garay2017bitcoin,garay2020full}, we define a notation of `good' epochs as follows. By abuse of notation, we write the expected PoW block rate $f^{(w)}$ simply as $f$ in the analysis below.

\begin{defn}
\label{def:good}
Epoch $e$ is \emph{good} if $(1-\delta)\alpha_0 f/\gamma \leq p h^{(e)} T_e \leq (1+\delta) \gamma f$, where $p = 1/2^\kappa$ and $h^{(e)} = h_{(e-1)R+1}$, i.e., the number of honest queries in the first slot of epoch $e$. A slot $r$ is good if it is in a good epoch. 
\end{defn}

\begin{prop}
\label{prop:good_slot}
Under a $(\gamma, 2R)$-respecting environment, if slot $r$ is a good slot in epoch $e$, then $(1-\delta)\alpha_0 f/\gamma^2 \leq p h_r T_e \leq (1+\delta)\gamma^2 f$.
\end{prop}

\begin{proof}
By the definition of a good epoch, we have $(1-\delta) \alpha_0 f/\gamma \leq p h^{(e)} T_e \leq (1+\delta) \gamma f$. And under a $(\gamma, 2R)$-respecting environment and the fact that slot $r$ is in epoch $e$, we know $1/\gamma \leq h_r/h^{(e)} \leq \gamma$. Therefore, $(1-\delta)\alpha_0 f/\gamma^2 \leq p h_r T_e \leq (1+\delta)\gamma^2 f$.
\end{proof}

Now, we prove {\it fairness} of PoW blocks in any large enough window, i.e., a miner controlling a $\phi$ fraction of the computational resources will contribute a $\phi$ fraction of work. We first give a form definition of fairness.

\begin{defn}
We say that $\H$ is a \emph{$\phi$-fraction honest subset} if miners in $\H$ (that may change over time) are honest and $n^{\H}_r > \phi n_r$ for any slot $r$, where $n^{\H}_r$ is number of queries in $\H$ at slot $r$.
\end{defn}

\begin{defn}[Fairness] 
\label{def:fairness}%
We say that the protocol has \emph{(approximate) $(W_0,\sigma)-$fairness} if, for any $\phi$, $\alpha_0 \leq \phi < 1$, %
any $\phi-$fraction honest subset $\H$, and any honest miner holding chain $\mathcal{C}$ at slot $r$ and any interval $S_0 \subseteq [0,r]$ with at least $W_0$ consecutive slots, 
it holds that the PoW blocks included in $\mathcal{C}(S_0)$ mined by $\H$ have total difficulty at least $(1-\sigma) \phi d$, where
$d$ is the total difficulty of all PoW blocks included in $\mathcal{C}(S_0)$.
\end{defn}

Define the random variable $D_r$ equal to the sum of the difficulties of all PoW blocks computed by honest miners at slot $r$. And for fixed $\phi$ and $\phi-$fraction honest subset $\H$, define the random variable $D^{\H}_r$ equal to the sum of the difficulties of all PoW blocks computed by miners in $\H$ at slot $r$. For a set of slots $S$, we define $D(S) = \sum_{r\in S}D_r$, $D^{\H}(S)= \sum_{r\in S}D^{\H}_r$, and $n^{\H}(S)= \sum_{r\in S}n^{\H}_r$. For a set of $J$ adversarial queries, define the random variable $A(J)$, as the sum of difficulties of all the PoW blocks created during queries in $J$. 

Next we define the notion of typical executions, which will be shown to occur with overwhelming probability.
\begin{defn}[Typical execution]
\label{def:typicality}
An execution $E$ is \emph{typical} if the following hold

\noindent(a) For any set $S$ of at least $\ell$ consecutive good slots,
\begin{align*}
     D(S) < (1+\epsilon)ph(S).
\end{align*}

\noindent(b) For any set $S$ of at least $\ell$ consecutive good slots,
let $J$ be the set of adversarial queries in $S$. If we further know each query in $J$ made at slot $r$ with target $T$ satisfies $(1-\delta)\alpha_0 f/\gamma^2 \leq p n_r T \leq (1+\delta)\gamma^2 f/\alpha_0$, then
\begin{equation*}
     D(S) + A(J) < (1+\epsilon)p(h(S)+|J|).
\end{equation*}

\noindent(c) For any set $S$ of at least $\ell/\phi$ consecutive good slots,
\begin{equation*}
    D^{\H}(S) > (1-\epsilon)pn^{\H}(S).
\end{equation*}
\end{defn}

To obtain meaningful concentration, we should be considering a sufficiently long slot sequence of at least
\begin{equation*}
    \ell \triangleq \frac{2(1+\epsilon/3)}{\epsilon^2\gamma^3(1-\delta)\alpha_0 f} \lambda,
\end{equation*}
where $\lambda$ is called the typicality parameter of the execution.

For our analysis to go through, the protocol parameters should satisfy certain conditions which we now discuss. First, we will require that the number $\ell$ defined above and the security parameter $\kappa$ are appropriately small compared to $R$, the duration of an epoch.
\begin{equation}
    \tag{C1}
    R - 3\kappa - \Delta \geq \ell/\alpha_0 \geq \frac{\gamma}{\epsilon}(4\kappa + \Delta).
\end{equation}
Note that (C1) implies $R \geq (4\kappa+\Delta)/\epsilon$. Second, the slack variables $\epsilon$ and $\delta$ should satisfy
\begin{equation}
    \tag{C2}
    4\epsilon \leq \delta \leq 1.
\end{equation}

Next we bound the probability of an atypical execution (Lemma~\ref{lem:typical}) and show that all epochs are good in a typical execution by induction (Lemma~\ref{lem:good}). Therefore, all epochs are good with overwhelming probability.

\begin{lemma}
\label{lem:typical}
For an execution $\mathcal{E}$ of $LR$ slots, in a $(\gamma,2R)$-respecting environment, the probability of the event  ``$\mathcal{E}$ is not typical'' is bounded by $\mathcal{O}(LR)e^{-\lambda}$.
\end{lemma}

\begin{lemma}
\label{lem:good}
For a typical execution in a $(\gamma, 2R)$-respecting environment, all epochs are good.
\end{lemma}

The proofs of Lemma~\ref{lem:typical} and Lemma~\ref{lem:good} are in Appendix~\ref{app:proof}. 
In order to address the dependency among PoW successes in the variable mining difficulty setting, the proof of Lemma~\ref{lem:typical} uses martingale arguments (from~\cite{garay2020full}) to provide useful bounds on the relevant random variables. The proof of Lemma~\ref{lem:good} uses an induction argument to show that if all previous epochs are good, then the progress of the adversary in the next epoch will be bounded and as a result the next epoch is also good. 

Finally, we prove the \fruitchainsnosp/fairness argument.
\begin{theorem}[Fairness]
\label{thm:fairness}
For a typical execution in a $(\gamma, 2R)$-respecting environment, the protocol with recency parameter $sl_{\rm re} = 3\kappa+\Delta$ satisfies $(W_0, \sigma)-$fairness, where $W_0 = \ell/\alpha_0 + 3\kappa +\Delta $ and $\sigma = 4\epsilon $.
\end{theorem}

\begin{proof}
Fix $\phi$, a $\phi-$fraction honest subset $\H$, and an honest node holding chain $\mathcal{C}$.  
Let $S_0 = \{u:r_1\leq u \leq r_2\}$ be a window of at least $W_0$ consecutive slots. Let $\mathcal{C}(S_0)$ be the segment of $\mathcal{C}$ containing PoS blocks with timestamps in $S_0$, let $\B$ be all PoW blocks included in $\mathcal{C}(S_0)$, and $d$ be the total difficulty of all PoW blocks in $\B$. 
The proof of Lemma~\ref{lem:good} observes the following facts to be used below:
\begin{itemize}
\item{Fact 1.} For any PoW block $B \in \B$, $B$ is mined after $r_1 - 4\kappa - \Delta$.
\item {Fact 2.} For any PoW block $B \in \B$, $B$ is mined before $r_2 + \kappa$. \item {Fact 3.} If a PoW block $B$ is mined by $\H$ after $r_1$ and before $r_2 - 3\kappa -\Delta$, then $B \in \B$.
\end{itemize}
Now, let $S_1 = \{u: r_1-(4\kappa + \Delta) \leq u \leq r_2 + \kappa \}$, $S_2 = \{u: r_1 \leq u \leq r_2 - (3\kappa +\Delta)\}$, and $J$ be the set of adversary queries associated with $\B$ in $S_1$. Then by Fact 1 and Fact 2, we have that all PoW blocks in $\B$ are mined in $S_1$; by Fact 3, we have that all PoW blocks mined by $\H$ in $S_2$ are in $\B$. 
Similar to the arguments in Lemma~\ref{lem:good}, for each query in $S_1$ made by either an honest node or the adversary at slot $r$ in epoch $e$, the target $T$ must satisfy $(1-\delta)\alpha_0 f/\gamma^2 \leq p n_r T \leq (1+\delta)\gamma^2 f/\alpha_0$.

Further note that, to prove $\sigma$-fairness, it suffices to show that
\begin{equation*}
    D^{\H}(S_2) \geq (1-\sigma)\phi(D(S_1)+A(J)).
\end{equation*}

Under a typical execution, we have
\begin{equation*}
    D^{\H}(S_2) > (1-\epsilon)pn^{\H}(S_2) \geq (1-\epsilon)\phi pn(S_2),
\end{equation*}
and
\begin{equation*}
    D(S_1) + A(J) < (1+\epsilon) p(h(S_1)+|J|) = (1+\epsilon) pn(S_1).
\end{equation*}

By our choice of $W_0$, we have $|S_2| \geq \ell/\alpha_0 \geq \ell/\phi$. Furthermore, we may assume $|S_2| \leq 2R$. This is because we may partition $S_2$ in parts such that each part has size between $\ell/\alpha_0$ and $2R$, sum over all parts to obtain the desired bound. Then by Proposition~\ref{prop:envbounds}, we have
\begin{align*}
    n(S_1) &\leq (1+\frac{\gamma|S_1 \setminus S_2|}{|S_2|})n(S_2) \\ 
    & \leq (1+\frac{\gamma(8\kappa + 2\Delta)}{\ell/\alpha_0})n(S_2) \\ 
    & \overset{(C1)}{\leq} (1+2\epsilon)n(S_2).
\end{align*}
Finally, by setting $\sigma = 4\epsilon$, we conclude the proof.
\end{proof}

\subsubsection{Lifting argument from single-epoch to multiple-epoch}


Theorem~\ref{thm:single} gives bounds for CP and $\exists$CQ for a single-epoch run of the protocol with static stake distribution and perfect randomness. We now conclude our proof of Theorem~\ref{thm:main1} by showing that these blockchain properties hold throughout the whole lifetime of the system consisting of many epochs.

\begin{theorem}[Restatement of Theorem~\ref{thm:main1}]
\label{thm:re_main}
Fix parameters $\kappa$, $\lambda$, $\ell$, $\gamma$, $\sigma$, $\epsilon$, $\delta$, $R$ and $L$ satisfying conditions (C1) and (C2). Under Assumption~\ref{ass:main1}, Assumption~\ref{ass:main2} and assumptions (P1)-(P3), when executed in a lock-step synchronous model (i.e., $\Delta$-synchronous model with $\Delta = 1$), \minotaur constructed with Ouroboros generates a transaction ledger that satisfies {\it persistence} and {\it liveness}  throughout a period of $L$ epochs (each with $R$ slots) with probability $1 - RL(e^{-\Omega(\kappa)}+e^{-\lambda})$.
\end{theorem}

This part of the analysis proceeds similarly as in Section 5 of~\cite{kiayias2017ouroboros} and hence we only sketch it in Appendix~\ref{app:proof}.
Applying the improved analysis in~\cite{bkmqr2020} (instead of the one in~\cite{kiayias2017ouroboros}), we obtain the refined bound
$1 - RL(e^{-\Omega(\delta^3\kappa)}+e^{-\lambda})$ for the statement of Theorem~\ref{thm:re_main}.
In Appendix~\ref{app:tight}, we analyze a private-chain attack and show a corresponding lower bound.
It reveals that $\delta^2\kappa$ needs to be bounded below by a constant and, consequently,
the dependency on $\delta^3$ cannot be improved much further.

\subsection{Security with Ouroboros Praos/Genesis}

By {\sf maxvalid-bg}, an honest node prefers longer chains, if the new chain $\mathcal{C}'$ forks less than $k$ blocks relative to the currently held chain $\mathcal{C}$. If $\mathcal{C}'$ did fork more than $k$ blocks relative to $\mathcal{C}, \mathcal{C}'$ would still be preferred if it grows more quickly in the $s$ slots following the slot associated with the last common block of $\mathcal{C}$ and $\mathcal{C}'$. In this section, we analyze the security of the protocol with {\sf maxvalid-bg}, but drop assumption (P3) (i.e., allow bootstrapping from genesis).

\begin{proof}[Proof sketch] This part is identical to the analysis in Ouroboros Genesis (Section 4.3-4.4 of~\cite{badertscher2018ouroboros}), so we only sketch it here. The proof proceeds in two steps:
\begin{enumerate}
    \item[1)] For an honest node $h_1$ that is always online, we show that when replacing {\sf maxvalid-mc} with {\sf maxvalid-bg}, the overall execution of the protocol in $h_1$'s view remains the same except with negligible probability. That is to show that with overwhelming probability, whenever $h_1$ receives a new chain $\mathcal{C}'$ that forks more than $k$ blocks from $h_1$'s local chain $\mathcal{C}$, $h_1$ will always favor $\mathcal{C}$, i.e., $\mathcal{C}$ grows more quickly than $\mathcal{C}'$ in the first $s$ slots right after the fork.
    \item[2)] For a newly joined node $h_2$, we consider a `virtual' node $h_2'$ that holds no stake, but was participating in the protocol since the beginning and was honest all the time. Then we show $h_2$ will always adopt the same chains as $h_2'$ after joining the network.\qedhere
\end{enumerate}%
\end{proof}

\subsection{Comparison to Bitcoin}
\label{sec:fluctuation analysis}

In this paragraph, we observe that \minotaur (when executed in the pure PoW case, $\omega=1$) is more robust against fluctuations in PoW participation than \bitcoinnosp.
In particular, we demonstrate that \bitcoin may not enjoy liveness if the number of
parties is allowed to halve every two weeks---while \minotaur does.
The theoretical results of this section are accompanied by experiments;
see \S\ref{sec:ratedrop} and Figure~\ref{fig:variable_diff}. 
\revadd{We point out that although \bitcoin mining power has gone up exponentially since its inception, other PoW blockchains that use the same difficulty adjustment rule may experience such drops in mining power.\footnote{For example, ETC dropped from 19 TH/s to 1.5 TH/s in 7 months period \url{https://2miners.com/etc-network-hashrate}, but more impressively ERGO dropped from 23.65 to 0.5 TH/s in just 2 weeks' time \url{https://2miners.com/erg-network-hashrate}.}
 }

Consider \bitcoin for the case that, initially, there are $n$ honest parties and
that the target is $T$. Using the model and notation of \cite{garay2017bitcoin} with $q=1$,
$m$ is the number of blocks in an epoch and $p=1/2^\kappa$. Thus, each slot is
successful with probability $f=1-(1-pT)^n\approx pTn$.
We will show that for $s=\lfloor m/pTn\rfloor$, liveness may fail in
a $(2,s)$-respecting environment.
In such an environment, the adversary is allowed to halve the number of
parties every $s\approx m/f$ slots (note that $m/f$ is the expected duration
of an epoch, which for \bitcoin is two weeks).

The attack proceeds in stages of $s$ slots, during which the adversary does
not mine nor delays messages.
Thus, in the beginning of each slot all parties have a chain of the same
length.
The adversary halves the number of parties at the beginning of each stage.
Let $N_i$ denote the number of honest parties in stage $i$, with $N_0=n$ and
$N_i=n/2^i$ for $i>0$.
The expected number of successful slots in stage $i$ is
\[
	s\cdot[1-(1-pT)^{N_i}]
	<spTN_i
	=\frac{spTn}{2^i}
	\le\frac m{2^i}
.\]

Consider the stage $j$ for which 
\[
	k-1<\frac m{2^j}\le2(k-1)
.\]
The expected number of blocks computed in the first $j$ stages is 
\[
	\frac m2+\frac m{2^2}+\cdots+\frac m{2^{j}}
	=m-\frac m{2^{j}}
	<m-k+1
\]
and at most $k-1$ in stage $j+1$.
Recalling that the median of a binomial distribution with mean $\mu$ is at
most $\lceil\mu\rceil$, we obtain that with probability at least $1/4$ at most
$m$ blocks have been computed by the end of stage $j+1$ and at most $k-1$ of
them have been computed in the last stage.
Assuming the liveness parameter $u<s$, any transaction provided to all honest parties for the
first $u$ slots of the final stage, will be in depth less than $k$. 
It follows that liveness does not hold in the final stage of the attack.

\begin{proposition}
\bitcoinnosp's ledger does not satisfy liveness in a $(2,s)$-respecting environment, for any $u<s$ and $s<m/pTn$.
\end{proposition}

In contrast, \minotaur does not suffer from such fluctuation. Intuitively,
this is because it inherits its security from the underlying proof-of-stake
protocol. In particular, the epochs are of fixed duration, implying that target
recalculation occurs regularly. This is evident from Condition~C1, which
allows much greater values for $R$ and $\gamma$.



\section{Experiments}
\label{sec:exp}







We implemented a prototype instantiation of a \minotaur client in Rust and the code can be found at \cite{Minotaurcode}. \revreplace{We also implemented a \bitcoin client as a benchmark.}{We also implemented \bitcoinnosp, \fruitchains and Ouroboros clients as benchmarks. Particularly, for better comparison we implemented the variable difficulty version of \fruitchains recently proposed in~\cite{wang2021securing}, which uses the same difficulty adjustment as \bitcoinnosp.} In this section, we briefly describe the architecture of our implementation and present experimental results to evaluate the concrete performance of \minotaur under different scenarios.

\subsection{Architecture}
\label{sec:arch}

We implemented the Input-Endorsers/\fruitchains variant of \minotaurnosp, where transactions are exclusively included in the PoW blocks. In this way, \minotaur can potentially achieve optimal throughput up to the capacity of the network as shown by the implementation of Prism~\cite{yang2019prism}, a pure PoW consensus protocol with similar blockchain structure. The system architecture is illustrated in Figure~\ref{fig:system-architecture}. Functionally it can be divided into the following three modules:

\begin{enumerate}
    \item \textit{Block Structure Manager}, which maintains the client's view of the blockchain, and communicates with peers to exchange new blocks.
    \item \textit{Miner}, which assembles new PoW blocks.
    \item \textit{Staker}, which assembles new PoS blocks.
\end{enumerate}

\noindent 
The goal of a \minotaur client is to maintain up-to-date information of the blockchain. They are stored in the following three data structures:

\begin{enumerate}
    \item \textit{Block Structure Database}, stores the graph structure of the blockchain (i.e., a directed acyclic graph (DAG) formed with PoW and PoS blocks).
    \item \textit{Memory Pool}, stores the set of transactions that have not been mined into any PoW block.
    \item \textit{PoW Block Pool}, stores the set of PoW blocks that have not been referenced by any PoS block.
\end{enumerate}


At the core of the \textbf{Block Structure Manager} is an \textit{event loop} which sends and receives network messages to/from peers, and a \textit{worker thread pool} which handles those messages. When a new block arrives, the worker thread first checks its proof-of-work/proof-of-stake, and if valid, then proceeds to relay the block to peers that have not received it. Next, the worker thread checks whether all blocks referenced by the new block, e.g. its parent, are already present in the Block Structure Database. If not, it buffers the block in an in-memory queue and defers further processing until all the referenced blocks have been received. Finally, the worker validates the block (e.g., verifying transaction signatures), and inserts the block into the Block Structure Database. If the block is a PoW block, the Block Structure Manager checks the Memory Pool against the transactions in this new block and removes any duplicates or conflicts from the Memory Pool, and also puts the block into the PoW Block Pool.

The \textbf{Miner} module assembles new PoW blocks. It is implemented as a busy-spinning loop. At the start of each round, it polls the Block Structure Database and the Memory Pool to update the block it is mining. When a new PoW block is mined, it will be inserted into the Block Structure Database, then sent to peers by the Block Structure Manager. The Memory Pool and PoW Block Pool will also be updated accordingly.

The \textbf{Staker} module works similarly, but assembles new PoS blocks. At the start of each round, it polls the Block Structure Database and the PoW Block Pool to update the block it is assembling. When a new PoS block is generated, it will be inserted into the Block Structure Database, then sent to peers by the Block Structure Manager. The PoW Block Pool will also be updated accordingly.

\begin{figure}
    \centering
    \resizebox{0.7\columnwidth}{!}{\tikzstyle{ledger} = [draw, fill=blue!20, rectangle, 
    minimum height=4em, minimum width=6em, text centered, text width=5em]
\tikzstyle{blockchain} = [draw, fill=red!20, rectangle, minimum height=4em, minimum width=6em, text centered, text width=5em]
\tikzstyle{miner} = [draw, fill=green!20, rectangle, minimum height=4em, minimum width=6em, text centered, text width=5em]

\tikzstyle{database} = [draw, fill=yellow!20, rectangle, rounded corners, text width=5em, minimum height=3em, minimum width=6em, text centered]

\begin{tikzpicture}[auto, node distance=2.8cm,>=latex']
    \node [database] (blockchaindb) {Block Structure Database};
    \node [blockchain, above=1cm of blockchaindb] (worker) {Block Structure manager};
    \node [ledger, left of=blockchaindb] (staker) {Staker};
    \node [miner, right of=blockchaindb] (miner) {Miner};
    \node [database, right of=worker] (mempool) {Memory Pool};
    \node [database, left of=worker] (tranpool) {PoW Block Pool};
    \node [above=0.5cm of worker] (newblk) {New Blocks};
    \node [right of=newblk] (newtx) {New Transactions};
    \draw [<->] (blockchaindb) -- node[name=a] {} (worker);
    \draw [<->] (blockchaindb) -- node[name=b] {} (miner);
    \draw [<->] (blockchaindb) -- node[name=c] {} (staker);
    \draw [<-] (worker) -- node[name=d] {} (miner);
    \draw [<-] (worker) -- node[name=d] {} (staker);
    \draw [<->] (miner) -- node[name=e]{} (mempool);
    \draw [<->] (miner) -- node[name=e]{} (tranpool);
    \draw [->] (worker) -- node[name=f]{} (mempool);
    \draw [->] (worker) -- node[name=g]{} (tranpool);
    \draw [<->] (staker) -- node[name=i]{} (tranpool);
    \draw [<->] (newblk) -- node[name=j]{} (worker);
    \draw [->] (newtx) -- node[name=k]{} (mempool);

\end{tikzpicture}}
    \caption{Architecture of our \minotaur client implementation.}
    \vspace{-0.5cm}
    \label{fig:system-architecture}
\end{figure}
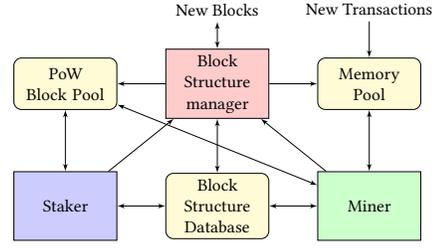

\subsection{Performance under variable mining power}\label{sec:ratedrop}


In \S\ref{sec:fluctuation analysis}, we have seen that \minotaur can survive more drastic variations of network hash power. We design the following experiment on our \minotaurnosp/\bitcoinnosp/\fruitchains codebase to verify this argument.

Recall that all three protocols have epoch-based target adjustment rules, varying the difficulty target of block mining based on the median inter-block time from the previous epoch. However, the epoch length in \minotaur is defined as a fixed number of slots instead of a fixed number of blocks (e.g., 2016 blocks in \bitcoinnosp). This definition makes sense in \minotaur because the main-chain blocks, which are PoS blocks, always have accurate timestamps (even if they are proposed by an adversary). In this experiment, we will see that this small change allows \minotaur to enjoy better liveness than \bitcoin and \fruitchains when the total mining power is decreasing.

In our experiment, we set the epoch length in \bitcoinnosp/\fruitchains to be 400 blocks and the expected duration of an epoch is 2 minutes; while the epoch length of \minotaur is 2 minutes and the expected number of blocks in a epoch is 400 blocks. Then \bitcoinnosp, \fruitchains and \minotaur are experiencing the same mining power variation, i.e., starting with an insufficient mining rate (100 blocks per minute) and then halving the mining rate every epoch (2 minutes in the experiments). From Figure~\ref{fig:variable_diff}, we can see that the mining target never has a chance to be adjusted in \bitcoinnosp/\fruitchains as the length of the main chain can never reach 400 blocks, which leads to a liveness failure; while the mining targets are adjusted every epoch (2 minutes) in \minotaur to make sure that the number of PoW blocks referenced by the main chain grows linearly. 

\begin{figure}
    \centering
    \includegraphics[width=0.4\textwidth]{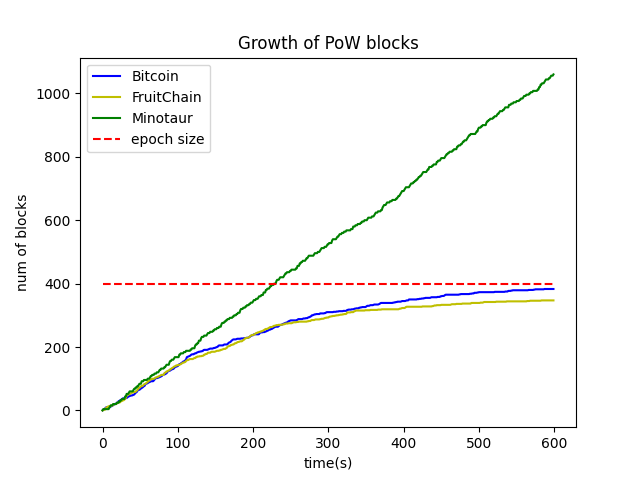}
    \vspace{-0.3cm}
    \caption{Behavior of \bitcoinnosp, \fruitchains and \minotaur under decreasing mining power. When the total mining power is halved every two minutes, the main chains of \bitcoin and \fruitchains stop growing as they can never reach the epoch size (400 blocks), while the number of PoW blocks referenced by the main chain grows linearly in \minotaurnosp.}
    \label{fig:variable_diff}
    \vspace{-0.4cm}
\end{figure}

\subsection{Performance under attacks}
In the following experiments, we evaluate how \minotaur performs in the presence of active attacks. Specifically, we consider three types of attacks: private chain, spamming and selfish mining attacks. \revadd{At a high level, a private chain attack aims to violate the persistence of the protocol.} Spamming attack aims to reduce network throughput, while selfish mining aims to compromise the PoW fairness of the protocol (i.e., reducing the work stake of honest miners).

\revadd[Request~\ref{req:revamp}\label{ans:privatechain}]{\noindent {\bf Private chain attack.} In a private chain attack, an attacker tries to privately generate an alternate chain faster than the public honest chain to displace a confirmed block~\cite{nakamoto2008bitcoin} (e.g., a block buried $k$-deep in the longest chain). In a pure PoW/PoS protocol (e.g., \bitcoinnosp/Ouroboros), an attacker with majority of mining/stake power can always succeed with the private chain attack, no matter how large the confirmation depth is. In contrast, \minotaur can survive such an attack, as long as the honest players control a majority of the combined resources.

To verify this key security property, we implemented the private chain attack on our code base, where an attacker can have various combinations of stake and mining power. The experiments of the private chain attack are as follows: the attacker (with $\beta_s$ fraction of stake and $\beta_w$ fraction of mining power) tries to generate a private chain starting at a specific block (e.g., the first block in some epoch) to overwrite the honest public chain when the private chain becomes longer. We run the attack for one epoch and count the length of the longest private chain (pruning honest blocks in the prefix if there is any) whose length is no less than that of the honest chain generated at the same time. Note that the maximum length of the private chain in a successful attack reflects the minimum confirmation depth that is needed to guarantee security of the protocol. We chose $\omega=0.5$ with a total PoS block rate being 1.58 block/s and a total PoW block rate being 2.70 block/s for \minotaurnosp. 
Parameters $\beta_s$ and $\beta_w$ are both iterated in $[0,1]$ with a step size of $0.1$ for \minotaur experiments, and we present the results as a 2D heatmap in Figure~\ref{fig:private-heatmap}. We also launched the same attack on Ouroboros Praos (\bitcoin resp.), where the attacker controls $\beta_s$ fraction of stake ($\beta_w$ fraction of mining power resp.) and the results are shown as 1D heatmaps in Figure~\ref{fig:private-heatmap}. 
In these experiments, block generation is set to the same rate as PoS block rate in \minotaur experiments. For each data point, we repeated the experiments 10 times and took an average to minimize randomness effects. From Figure~\ref{fig:private-heatmap} we observe that as long as $\beta_s+\beta_w<1$, the attacker can only succeed with very short private chain even when either $\beta_s$ or $\beta_w$ is close to 1, while Ouroboros Praos and \bitcoin would need $\beta_s < 0.5$ and $\beta_w < 0.5$ respectively. This provides empirical evidence that \minotaur is better than pure PoW or PoS protocols in term of security.

\begin{figure}
    \centering
    \includegraphics[width=0.8\columnwidth]{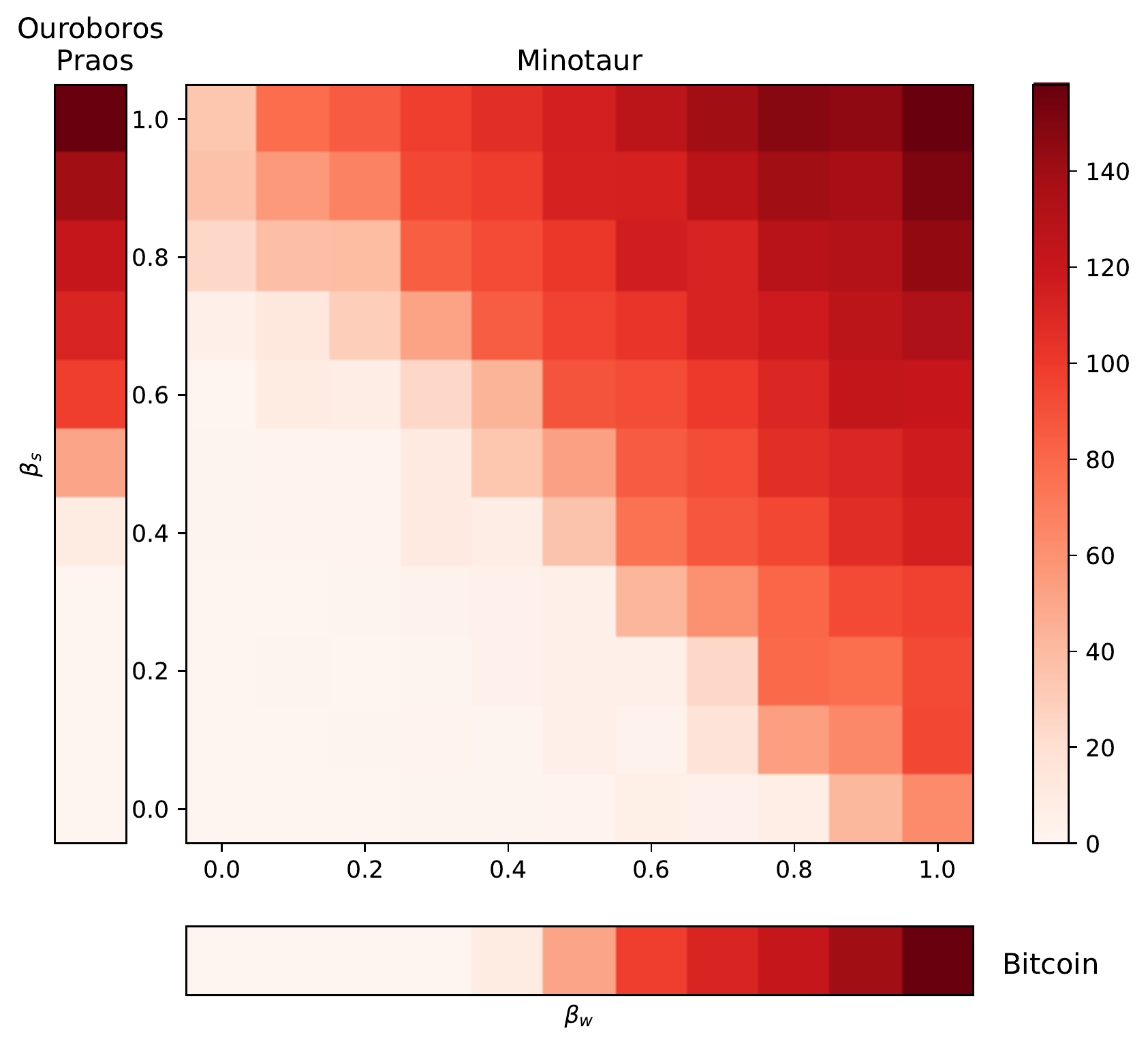}
    \vspace{-0.3cm}
    \caption{Longest private chain in private attacks with various adversarial stake fraction $\beta_s$ and adversarial mining power fraction $\beta_w$ on \minotaurnosp, Ouroboros Praos, and \bitcoinnosp. When $\beta_s + \beta_w < 1$ and either $\beta_s$ or $\beta_w$ is close to 1, the private chain in \minotaur is way shorter, meaning that it remains secure whereas Ouroboros Praos and \bitcoin are insecure.}
    \label{fig:private-heatmap}
    \vspace{-0.4cm}
\end{figure}

}

\noindent {\bf Spamming attack.} In a spamming attack, attackers flood the network with conflicting transactions. Among a bunch of conflicting transactions, at most one of them could be executed successfully. For example, in an account model, transactions with identical sender and nonce are conflicting transactions. Compared to \bitcoinnosp, the Input-Endorsers/\fruitchains variant of \minotaur cannot ensure that no conflicting transactions enter the transaction ledger since PoW are mined \textit{before} picked by main chain. 
Therefore, goodput (i.e., throughput of successfully executed transactions) of \minotaur transaction ledger will be comprised if conflicting transactions enter PoW blocks. To mitigate this attack, we propose a spam filter that operates as follows: a miner validates transactions with respect to the transaction ledger, PoW blocks that haven't been picked by main chain, and preceding pending transactions of this miner. In this way, if one of the conflicting transactions enters a PoW block and propagated throughout the network before other conflicting ones enter a PoW block, miners will filter out other conflicting transactions and increase goodput.

We implemented the spam filter and did experiments with four nodes in a line topology. Our metrics is \textit{normalized spam}, which is the count of unsuccessfully executed transactions normalized by the count of transactions in ledger. In our experiments, we make transaction generators create conflicting transactions on purpose, and send them to two nodes sitting at the ends of the line topology. And other nodes receive non-conflicting transactions. Hence, without any spam precaution, normalized spam is $0.25$. We set the block generation rate to be $0.44$ block/s for PoS blocks and $1.39$ block/s for PoW blocks. The results of our experiments with additional peer-to-peer latency is shown in Figure~\ref{fig:spam1}, and we observe that although normalized spam increases as additional latency increases, it is still lower than that without spam filter. The highest spam in our experiments is $0.063$, far lower than $0.25$. Noticed that in that experiment we apply an additional peer-to-peer latency of $0.3$s, which is larger than most peer-to-peer latency in \bitcoin and Ethereum~\cite{kim2018measuring}. This means that our spam filter can reduce spam to a minimal effect.

\begin{figure}
    \centering
    \includegraphics[width=0.8\columnwidth]{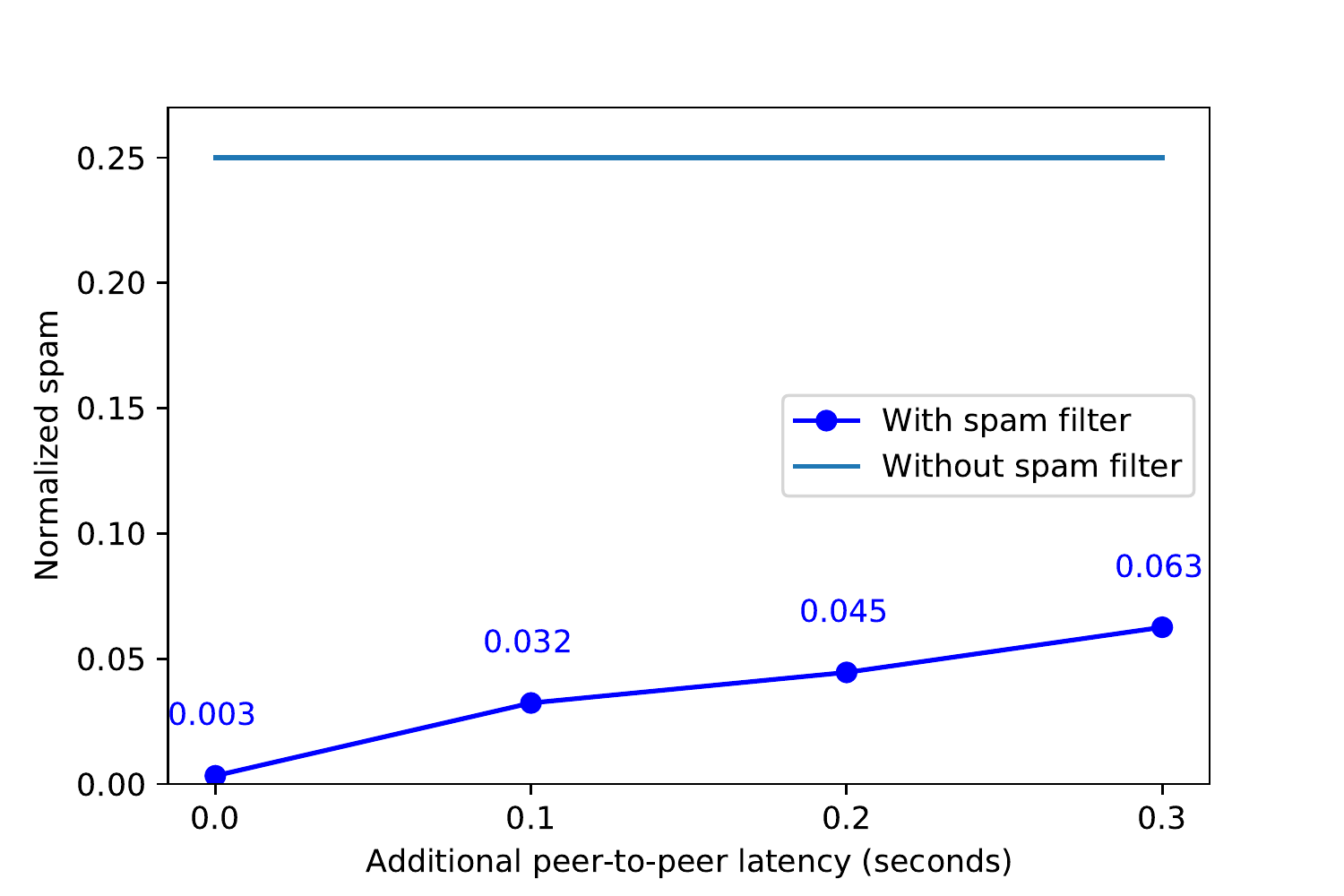}
    \vspace{-0.3cm}
    \caption{Normalized spam in 4-node experiments with additional peer-to-peer latency. Spam is reduced mostly under various additional latencies.}
    \label{fig:spam1}
    \vspace{-0.4cm}
\end{figure}

\ignore{ 
\mf{I am skeptical about the general effectiveness of this spam filter under arbitrary adversarial attacks.}

\xw{We cannot simulate all attacks. But I think there is at least no harm to add this spam filter.}

\mf{Yes, it does not hurt.}

\mf{(Comment now at the right place :-). This comparison is unfair against Minotaur, as we can run Minotaur by including the transacti<ons in the PoS blocks---just as in Bitcoin. Actually, I already redacted all previous statements in the paper that claimed that transactions go into PoW blocks for a different reason: we would also want PoS transaction (/input) blocks to be produced in order to achieve fairness and good throughput.}

\xw{At the beginning of section~\ref{sec:arch}, I mentioned that we implemented the Endorser/Fruitchains variant of \minotaur to enjoy optimal throughput.}
}

\noindent {\bf Selfish mining attack.} 
It has been known that \bitcoin is vulnerable to the selfish mining attack~\cite{nayak2016stubborn,sapirshtein2016optimal,eyal2014majority}, where a selfish miner withhold its mined blocks and release them later at an appropriate time to take the place of honest blocks in the longest chain. This attack hurts the fairness of the protocol, in the sense that the selfish miner with $\beta_w$ fraction of mining power can have more than $\beta_w$ faction of blocks in the main chain so that it will reap higher revenue. The \fruitchains protocol~\cite{pass2017analysis} was proposed as a solution to selfish mining. We observe that the selfish mining attack will not work on \minotaur Since it has similar blockchain structure as \fruitchainsnosp. \revadd[Request~\ref{req:fccomparison}\label{ans:fccomparison}]{Moreover, \minotaur will achieve fairness even when $\beta_w \geq 0.5$ (as long as there is an honest super-majority in actual stake), while \fruitchains would fail.} To verify these observations, we implemented the following selfish mining attacks on our \bitcoinnosp/\fruitchainsnosp/\minotaur code base:
\begin{itemize}[leftmargin=6mm]
    \item {\bf On \bitcoinnosp:} The selfish miner always mines on the block at the tip of the longest chain, whether the chain is private or public. Upon successful mining, the adversary maintains the block in private to release it at an appropriate time. In particular, when an honest miner publishes a block, the selfish miner will release a previously mined block at the same level (if it has one).
    \item \revadd[Request~\ref{req:fccomparison}\label{ans:fccomparison:2}]{{\bf On \fruitchainsnosp:} The selfish miner plays the same withholding-and-releasing strategy on its main-chain blocks as above, while its main-chain blocks contain fruits mined by itself exclusively. }
    \item {\bf On \minotaurnosp:} The selfish miner plays the same withholding-and-releasing strategy on its PoS blocks as above, while its PoS blocks contain PoW blocks mined by itself exclusively. 
\end{itemize}

We assume honest nodes will choose the attacker's block with probability $p$ whenever there is a tie. In our experiments, we set $p > 0.5$ to simulate the case that the attacker has better network connection than honest nodes so its blocks are usually transmitted faster.
Table~\ref{tab:selfish_mining} presents fractions of PoW blocks (or fruits) in the main chain of \bitcoinnosp/\fruitchainsnosp/\minotaur under various adversarial power (combinations of $\beta_w$ and $p$). On \minotaurnosp, the attacker controls $1/3$ of the virtual stake, i.e., $\omega \beta_w + (1-\omega) \beta_s = 1/3$, where $\beta_s$ is the fraction of the actual stake controlled by the attacker. In Table~\ref{tab:selfish_mining}, we can see that \bitcoin is indeed vulnerable to selfish mining, meaning that the attacker has more fraction of PoW blocks in the main chain (and thus more block rewards in \bitcoinnosp) than its fraction of mining power, particularly when $\beta_w \geq 0.5$; \revadd[Request~\ref{req:fccomparison}\label{ans:fccomparison:3}]{and \fruitchains only resists the attack when $\beta_w < 0.5$.} Meanwhile, in \minotaurnosp, the attacker can only have $\beta_w$ faction of PoW blocks in the main chain, no matter how much mining power it has (even when $\beta_w \geq 0.5$), and how the network favors it. This property is crucial to the security of \minotaur as we have seen in \S\ref{sec:fruits}.

\ignore{ 
\mf{I would consider analytical evidence for fairness by Fruitchains sufficient---not needing additional backing by experiments. Also, it seems that $\beta_s\neq 0$, which seems to make both settings incomparable.}

\xw{In our analysis, we implicitly assume that PoS blocks have unlimited size so that one honest PoS block can refer all unpicked PoW blocks. This is not true in practice. So the attacks are more like stress tests on our implementation. Also from my past experience, the CCS reviewers like to ask about selfish mining even if we already have full security analysis.

I set $\beta_s \neq 0$ to show the sharp comparison when $\beta_w < 0.5$. Indeed this seems unfair to Bitcoin. But this is the main advantage of Minotaur over Bitcoin.}
}

\begin{table}[]
    \centering
    \begin{tabular}{c|c|c|c|c|c|c}
    \hline
        \multicolumn{2}{c|}{Attacker's power} & \multicolumn{5}{c}{Fraction of PoW blocks} \\
    \hline
    \hline
        \multicolumn{2}{c|}{$\beta_w$} & 0.75 & 0.67 & 0.50 & 0.33 & 0.25  \\
    \hline
        \multirow{3}{4em}{$p=1$} & \bitcoin & 0 & 0 & 0.007 & 0.513 & 0.669  \\
    \cline{2-7}
    &  \fruitchains &   0	& 0.007 & 0.011 &	0.661 & 	0.739  \\
    \cline{2-7}
       & \minotaur & 0.248 &	0.332 & 0.498	& 0.665 &	0.746    \\
    \hline
        \multirow{3}{4em}{$p=0.7$} &  \bitcoin &  0 & 0.001 & 0.098 & 0.618 & 0.72 \\
    \cline{2-7}
    &  \fruitchains &  0	& 0.003 &	 0.056 &	0.665 & 	0.753  \\
    \cline{2-7}
       & \minotaur &  0.248 &	0.333	& 0.499 &	0.666	& 0.750   \\
    \hline
    \end{tabular}
    \caption{Fractions of honest PoW blocks (or fruits) in the main chain under selfish mining attacks in \bitcoinnosp/\fruitchainsnosp/\minotaurnosp. In these experiments, the attacker controls $\beta_w$ fraction of mining power and the tie breaking rule favors the attacker's block with probability $p$.}
    \vspace{-0.6cm}
    \label{tab:selfish_mining}
\end{table}

\section{Discussion}
\label{sec:discuss}

\noindent {\bf Variable weighing parameter.} In our security analysis, we assume a fixed weighing parameter $\omega$ through the whole execution of the protocol. However, we point out that $\omega$ can change over time as long as we put proper assumptions on the adversarial stake/mining power. Suppose $\omega(e)$ is the weighing parameter of epoch $e$, i.e., the ratio of work stake and actual stake in epoch $e$ is set to be $\omega(e)/(1-\omega(e))$. Note that the function $\omega(e)$ may be decided by the protocol designers and hard-coded in the genesis block, but the players can also reach an agreement (off-chain) to update it by doing a soft fork.

We give a brief security sketch, deferring the full analysis to future work.
Recall that, in our protocol, the virtual stake of a player in epoch $e$ composes of two parts: the actual stake drawn from the last PoS block in epoch $e-2$; and the work stake that is proportional to the amount of work it produced during epoch $e-2$. Thus, to guarantee the security in epoch $e$, we need the adversary to be $(1/2-2\sigma,2R,\omega(e))$-bounded (see Def.~\ref{def:adv}) in epoch $e-2$ and epoch $e-1$. And similarly, the adversary needs to be $(1/2-2\sigma,2R,\omega(e+1))$-bounded in epoch $e-1$ and epoch $e$. Therefore, the adversarial stake/mining power must be restricted by both above bounds for epoch $e-1$. As long as $\omega(e+1)$ does not differ too much from $\omega(e)$, this restriction on the adversary is reasonable to assume, and the weighing parameter can transit smoothly from $\omega(e)$ to $\omega(e+1)$. Figure~\ref{fig:transit} gives an example of how the assumptions shift when $\omega(e)$ is updated at the onset of epoch $e_0+1$.

\begin{figure}
    \centering
    \includegraphics[width=0.32\textwidth]{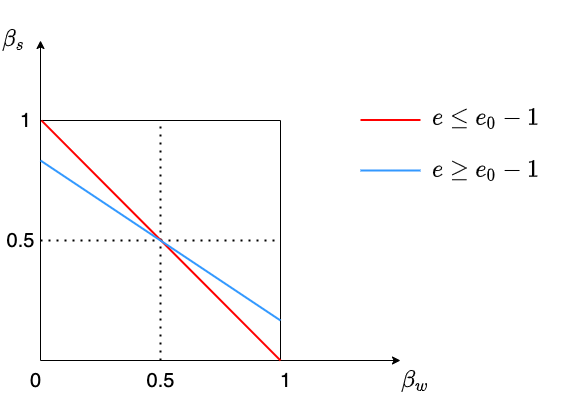}
    \caption{An example of variable weighing parameter $\omega$ in \minotaurnosp. Suppose the value of $\omega$ changes from $1/2$ to $2/5$ at the onset of epoch $e_0+1$. Then the adversarial stake/mining power needs to be restricted below the red line in epochs $e\leq e_0-1$ and below the blue line in epochs $e\geq e_0-1$. In particular, the adversary is restricted by both lines in epoch $e_0-1$. When $\omega$ changes smoothly over time, it is reasonable to assume exactly this.}
    \vspace{-0.5cm}
    \label{fig:transit}
\end{figure}

We point out that such flexible weighing between work and stake is very useful in practice. For example, in general, PoS blockchains are easy to launch with existing techniques, such as proof- of burn~\cite{karantias2020proof}, initial coin offering~\cite{li2018initial} and airdrop~\cite{airdrop}. Therefore, \minotaur can be launched as a pure PoS blockchain and later transit into a hybrid PoW/PoS one or a pure PoW blockchain.
In addition, the security of \minotaur can be enhanced by assigning higher weight to the more decentralized resource. 

\noindent {\bf Removing the initial constraint.}
As discussed in \S\ref{sec:prelim}, we  use an honest majority of stake without relying on PoW for the two initial epochs. 
We note that we can relax this initial constraint 
harmonizing even the initial two epochs with the rest of the execution by employing the smooth-interpolation technique  (\S\ref{sec:base}, Idea~2) where work blocks are mined based on $2$-for-$1$ PoW~\cite{garay2020full} producing both, main-chain blocks (at a rate defined by $\omega$) and endorser blocks (at a sufficient rate to guarantee fairness).
This allows to start the protocol with any initial weighing by means of $\omega$.
\revreplace[Request~\ref{req:initialstake}\label{ans:initialstake}]{We defer a formal analysis to future work.}{We do not pursue this further as, (I) for most practical purposes we expect the present construction  to be sufficient as is (namely it will be reasonable  to  launch \minotaur with PoW ``turned off'', at $\omega=0$ and then begin to incrementally adjust $\omega$ to higher values at some time after two epochs), and, (II) as explained in \S\ref{sec:base}, the analysis of the variation described above that allows $\omega>0$ from launch will follow in a straightforward manner from the forkable string analysis and input-endorser technique of \cite{kiayias2017ouroboros,badertscher2018ouroboros} as well as the fairness argument of \cite{pass2017fruitchains}. 
This is because the main difficulty in the analysis of PoW/PoS hybrid protocols is dealing with variable difficulty --- and adopting the $\omega=0$ convention for the initial two epochs allows us to focus on that. 
}
%


\noindent {\bf Generalization to multiple resources.}
Following the idea in~\cite{TNDF+18}, \minotaur can be extended to more than two resources. In contrast to their construction, this is achievable without fundamentally changing the structure of the protocol.

Assume $M\geq 1$ different resources, and for each such resource, an independent lottery mechanism among the contributors to assign `successes' proportionally to a party's ratio of the total contributed resource.
%
Defining a fixed amount of virtual stake $V$ and weights $\omega_i>0$,
$\sum_{i=1}^{M}\omega_i=1$, we have resource $i$ control $\omega_i\cdot V$ of
the virtual stake.
%
During each epoch $e$, the different lotteries are run concurrently, wherein each success allows for the release of a respective block (to be eventually picked up by a block of the main chain), tied to the respective resource.

Main-chain block leadership for epoch
$e+2$ is assigned to resource $i$ based on relative virtual stake $\omega_i$, and further split among the contributors of that resource based on their production
of respective `resource blocks' during epoch $e$.
%
%
Naturally, this protocol tolerates
\begin{equation}\label{eq:multi-resource}
 \sum_{i=1}^M \omega_i \beta_i < \frac{1}{2}\,,
\end{equation}
where $\beta_i\in[0,1]$ is the fraction of resource $i$ held by the adversary.
A pictorial example for $M=3$ is given in Figure~\ref{fig:multi-resource}.

Although the underlying blockchain protocol is PoS-based, stake (in the classical, non-virtual sense) is not required to be one of the $M$ resources, i.e., the protocol can be run solely on virtual stake. Also, full adversarial control of some of the resources can be tolerated as long as Eq.~(\ref{eq:multi-resource}) is satisfied.

Finally, note that the resources can also be of the same type, e.g., \minotaur can combine work emanating from different hash functions such as SHA$256$, scrypt, and ethash---answering an open question raised in \S\ref{sec:intro}.


\begin{figure}
\centering
\begin{tikzpicture}[scale=0.16]
 \def\sc{0.8} 
 \filldraw[black] (0,0) circle(2pt);
 \filldraw[black] (12,0) circle(2pt);
  \filldraw (12,-0.5) node[anchor=north] {$1$};
  \filldraw (12+1.5,-1) node[anchor=west] {$\beta_1\ [\omega_1=\frac{1}{4}]$};
 \filldraw[black] (0,12) circle(2pt);
  \filldraw (-0.5,12+0.5) node[anchor=east] {$1$};
  \filldraw (0,12+2.5) node[anchor=west] {$\beta_2\ [\omega_2=\frac{1}{3}]$};
 \filldraw[black] (-\sc*9,-\sc*6) circle(2pt);
  \filldraw (-\sc*9-0.5,-\sc*6+0.5) node[anchor=east] {$1$};
  \filldraw (-\sc*9-3.5,-\sc*6-3) node[anchor=west] {$\beta_3\ [\omega_3=\frac{5}{12}]$};
 \draw[->,gray,thin] (0,0) -- (0,12+1);
 \draw[->,gray,thin] (0,0) -- (12+1,0);
 \draw[->,gray,thin] (0,0) -- (-9*\sc-0.75,-6*\sc-0.5);

 \coordinate (t11) at (12,0);
 \coordinate (t12) at (12,9);
 \coordinate (t13) at (12-\sc*3/5*9,-\sc*3/5*6);

 \coordinate (t22) at (0,12);
 \coordinate (t21) at (8,12);
 \coordinate (t23) at (-\sc*2/5*9,12-\sc*2/5*6);

 \coordinate (t33) at (-\sc*9,-\sc*6);
 \coordinate (t31) at ({\sc*(-9+1/3*12)},-\sc*6);
 \coordinate (t32) at (-\sc*9,{\sc*(-6+1/4*12)});

 \filldraw[draw=black,fill=gray!5] (t11) -- (t12) -- (t13) -- cycle;
 \filldraw[draw=black,fill=gray!5] (t21) -- (t22) -- (t23) -- cycle;
 \filldraw[draw=black,fill=gray!5] (t31) -- (t32) -- (t33) -- cycle;

 \filldraw[draw=black,fill=gray!5] (t12) -- (t21) -- (t23) -- (t32) -- (t31) -- (t13) -- cycle;

 \draw[gray,thin,dashed] (0,0) -- (0,12);
 \draw[gray,thin,dashed] (0,0) -- (12,0);
 \draw[gray,thin,dashed] (0,0) -- (-9*\sc,-6*\sc);

 \draw[gray,thin,dotted] (0,12) -- (12,12); 
 \draw[gray,thin,dotted] (12,12) -- (12,0);
 \draw[gray,thin,dotted] (12,0) -- (12-\sc*9,0-\sc*6); 
 \draw[gray,thin,dotted] (12-\sc*9,0-\sc*6) -- (-\sc*9,-\sc*6);
 \draw[gray,thin,dotted] (-\sc*9,-\sc*6) -- (-\sc*9,-\sc*6+12); 
 \draw[gray,thin,dotted] (-\sc*9,-\sc*6+12) -- (0,12);
 \draw[gray,thin,dotted] (-\sc*9,-\sc*6+12) -- (-\sc*9+12,-\sc*6+12);
 \draw[gray,thin,dotted] (-{\sc*9+12},-\sc*6) -- (-{\sc*9+12},-\sc*6+12);
 \draw[gray,thin,dotted] (12,12) -- (-{\sc*9+12},-\sc*6+12);

 \ignore{
 \draw[blue,thin,dotted] (6-\sc*9/2,-\sc*6/2) -- (6-\sc*9/2,-\sc*6/2+6); 
\draw[blue,thin,dotted] (24,0) -- (0,18); %
\draw[blue,thin,dotted] (0,18) -- (-\sc*9*6/5,-\sc*6*6/5); %
\draw[blue,thin,dotted] (-\sc*9*6/5,-\sc*6*6/5) -- (24,0); %
\draw[blue,thin,dotted] (12,0) -- (24,0); %
\draw[blue,thin,dotted] (-\sc*9,-\sc*6) -- (-\sc*9*6/5,-\sc*6*6/5); %
\draw[blue,thin,dotted] (0,12) -- (0,18); %
}

\end{tikzpicture}
\vspace{-0.2cm}
\caption{Tolerable resource corruption for %
$\mathbf{\omega}=(\frac{1}{4},\frac{1}{3},\frac{5}{12})$. The volume enclosed by the gray surfaces is tolerated by the protocol: $\sum_{i=1}^3 \omega_i \beta_i <\frac{1}{2}$.} 
\vspace{-0.45cm}
\label{fig:multi-resource}
\end{figure}
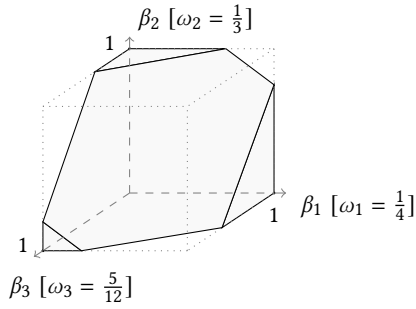

\noindent {\bf Fairness with respect to the combined resources.}
Fairness with respect to PoW blocks (Definition~\ref{def:fairness}) is a crucial property of \minotaur to guarantee fair assignment of work stake to the miners. However, this is not sufficient for aspects like fair reward sharing~\cite{pass2017fruitchains}, since, for this purpose, fairness needs to be achieved for overall block production with respect to the \emph{combined} resources.
We note that fairness with respect to the combined resources can be achieved by scheduling endorser blocks for every involved resource, \emph{including stake}. A given epoch reward can now be shared by assigning an $\omega_i$ fraction of the reward to the contributors of the $i$-th resource, and distributing each such fraction proportionally to the parties' contributions of endorser blocks for that resource.


\revadd[Request~\ref{req:comparison}\label{ans:comparison}]{
\noindent {\bf Comparison to state-of-the-art protocols.  }
In comparison to previous hybrid PoW-PoS protocols,~\cite{FDKT+20} already generalized the conditions under which a permissionless blockchain can be securely operated by combining PoW and PoS. Our paper improves over~\cite{FDKT+20} by demonstrating how to tolerate a strictly better bound, i.e., \emph{any} adversarial minority of the combined resources, and furthermore proving this bound tight, thus settling the hybrid PoW-PoS question. 

We also note that \minotaurnosp's performance is comparable to previous state-of-the-art protocols:

$\bullet$ \emph{State-of-the-art PoS}. As \minotaur resembles a PoS protocol (while operating on virtual stake derived from PoW), the only overhead it introduces over its underlying PoS protocol is the extra (insignificant) effort required for virtual-stake conversion---which impacts system throughput or transaction latency only marginally.
The performance of our explicit \minotaur construction, based on Ouroboros, is thus directly inherited from the Ouroboros protocol.

$\bullet$ \emph{State-of-the-art PoW}. Note that Nakamoto-style PoS mimics the workings of PoW Nakamoto consensus and thus yields similar stochastics as in the PoW case, and thus also comparable performance characteristics.
}


\section{Acknowledgements}
We would like to thank Thomas Kerber and Alexander Russell for insightful discussions.

This research is supported in part by the US National Science Foundation under grants CCF-1705007 and CNS-1718270 and the US Army Research Office under grant W911NF1810332.

\bibliographystyle{acm}
\bibliography{ref}

\appendices
\section*{Appendix}

\section{Security regions of hybrid protocols in Figure~1b}
\label{app:defend}

We give details of the security regions plotted in Figure~\ref{fig:pow_pos}. Recall that $\beta_w$ is the proportion of adversarial hash power and $\beta_s$ is the proportion of the adversarial stake (or the proportion of Byzantine nodes if a permissioned BFT protocol is adopted).

\noindent{\bf 2-hop blockchain.} We derived the security region from the assumptions in Theorem 2\&3 of~\cite{FDKT+20}. When $\Delta = 0$ (which gives the largest security region), the security proof of~\cite{FDKT+20} assumes $(1-\beta_w)(1-\beta_s) > \beta_s$, i.e., $\beta_s < \frac{1-\beta_w}{2-\beta_w}$ implying $\beta_s < 1/2$.

\noindent{\bf Checkpointed ledger.} \cite{karakostas2021securing} uses a synchronous BFT protocol (with $1/2$ fault tolerance) to regularly issuing {\em checkpoints} on a PoW longest chain. The protocol is proven to be safe and live when $\beta_s<1/2$.

\noindent{\bf Finality gadgets.} \cite{neu2021ebb,sankagiri2020blockchain,buterin2017casper,stewart2020grandpa} use an asynchronous/partial synchronous BFT protocol (with $1/3$ fault tolerance) to build a finality gadget/layer on the top of a PoW longest chain to achieve important properties such as finality (a.k.a deterministic safety under asynchrony) and accountability. These protocols are proven to be safe and live when $\beta_s < 1/3$ and $\beta_w < 1/2$.

\section{Mathematical facts}
\label{app:math}

\begin{theorem}[from \cite{garay2020full}]
\label{thm:martingale}
Let $(X_1,X_2,\ldots)$ be a martingale with respective the sequence $(Y_1,Y_2,\ldots)$, if an event $G$ implies $X_k - X_{k-1} \leq b$ and $V = \sum_{k}var[X_k - X_{k-1}| Y_1, \ldots, Y_{k-1}] \leq v$, then for non-negative $n$ and $t$
\begin{align*}
    P(X_n -X_0 \geq  t, G ) \leq e^{-\frac{t^2}{2v + \frac{2bt}{3}}}.
\end{align*}
\end{theorem}

The following is known as the Berry-Esseen Theorem. See \cite{feller2} as a standard reference and \cite{shevtsova} for improvements with respect to the constant 1/2.
\begin{theorem}
	Let the $X_i$ be independent variables with common distribution such that
	$\E[X_i]=\mu$, $\V[X_i]=\sigma^2>0$, $\E[|X_i-\mu|^3]=\rho<\infty$.
	If $F_n$ is the distribution of $(X_1+\cdots+X_n-\mu n)/\sqrt{n\sigma^2}$
	and $\Phi$ the standard normal, then
	\[
		|F_n(x)-\Phi(x)|\le\frac\rho{2\sigma^3\sqrt n},\quad
		\text{for all $x$ and $n$.}
	\]
\end{theorem}


\section{A long range attack}
\label{app:long_range}

We point out that \minotaur is insecure with the longest chain rule due to a long range attack. Let the weighing parameter $\omega = 0.5$ and define $f$ to be PoS block production rate by a stakeholder who controls all actual stake in the system. Suppose the adversary controls 0.8 fraction of stake and 0.1 fraction of mining power at some slot (not at the beginning of the execution), then after behaving honestly for some time, the adversary will control 0.8 fraction of actual stake and 0.1 fraction of work stake at the beginning of some epoch $e$. Now, the adversary starts to grow a private chain $\mathcal{C}_1$, while honest nodes grow a public chain $\mathcal{C}_2$. Suppose $\mathcal{C}_1$ and $\mathcal{C}_2$ include the same set of transactions. Then at the beginning of epoch $e+2$, the adversary will control 0.8 fraction of actual stake and all the work stake (after normalization) on its chain $\mathcal{C}_1$ because $\mathcal{C}_1$ only refers adversarial PoW blocks in epoch $e$. So the growth rate of $\mathcal{C}_1$ after epoch $e+2$ will be $0.8f +f = 1.8f$. Similarly, at the beginning of epoch $e+2$, honest nodes will control 0.2 fraction of actual stake and all the work stake (after normalization) on the chain $\mathcal{C}_2$ and the growth rate of $\mathcal{C}_2$ after epoch $e+2$ will be $0.2f +f = 1.2f$. Therefore, $\mathcal{C}_1$ will catch up with $\mathcal{C}_2$ eventually if the longest chain rule is adopted. See Figure~\ref{fig:long_range}.

This long range attack is similar to the stake grinding attack on Ouroboros and it can be prevented by the new chain selection rules {\sf maxvalid-mc} and {\sf maxvalid-bg}. For {\sf maxvalid-mc}, honest nodes won't accept $\mathcal{C}_1$ because it forks too long from $\mathcal{C}_2$; for {\sf maxvalid-bg}, honest nodes won't accept $\mathcal{C}_1$ because it grows slower than $\mathcal{C}_2$ in epoch $e$ (right after the fork). 

\begin{figure*}
    \centering
    \includegraphics[width=0.7\textwidth]{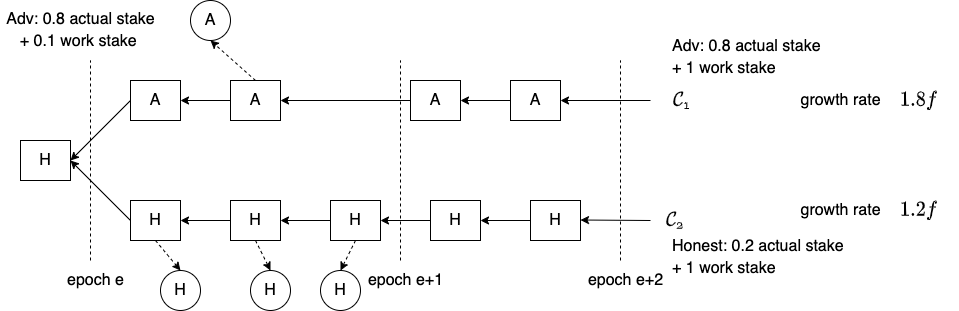}
    \caption{A long range attack on the longest chain rule.}
    \label{fig:long_range}
\end{figure*}

\section{Proof for \S6}
\label{app:proof}

\begin{proof}[Proof of Lemma~\ref{lem:typical}]
The proof for (a) is the same as part (a) of Theorem 1 in \cite{garay2020full}. 
For (b), by the condition, for each query in $S$ made by either an honest node or the adversary at slot $r$ in epoch $e$, the target $T$ must satisfy $(1-\delta)\alpha_0 f/\gamma^2 \leq p n_r T \leq (1+\delta)\gamma^2 f/\alpha_0$
Therefore, the proof is also similar to part (a) of Theorem 1 in \cite{garay2020full}. 

For (c), let the execution be partitioned into parts such that each part has at least $\ell/\phi$ and at most $s = 2R$ slots. We prove that the statement fails with a probability less than $e^{-\lambda}$ for each part. 
Let $J$ denote the queries made by $\mathcal{H}$ in slots $S$. We have $|J| =n^{\mathcal{H}}(S)= \phi n(S) $. For $k \in [|J|]$, let $Z_i$ be the difficulty of any block obtained from query $j \in J$ and we write $\mathcal{E}_{j-1}$ for the execution just before this query. Then
\begin{align*}
    &X_0 = 0 \\
    &X_k = \sum_{i \in [k]} Z_i - \sum_{i \in [k]} \mathbb{E}[Z_i| \mathcal{E}_{i-1}]  
\end{align*}
is a martingale with respect to $\mathcal{E}_0, \ldots, \mathcal{E}_{k}$. We have 
\begin{align*}
   X_k - X_{k-1} &= X_k - \mathbb{E}[X_k | \mathcal{E}_{k-1}] = Z_k - \mathbb{E}[Z_k | \mathcal{E}_{k-1}] \\
                 &\leq \frac{1}{T_k}  = \frac{p h_k}{p h_k T_k} \leq \gamma^3 ph(S)/(1-\delta)\alpha_0 f|S| \\
                 &\leq \gamma^3 pn(S)/(1-\delta)\alpha_0 f|S| := b.
\end{align*}
Similarly 
\begin{align*}
    V &= \sum_k var [X_k - X_{k-1} | \mathcal{E}_{k-1}] \leq \sum_k \mathbb{E}[Z_k^2 | \mathcal{E}_{k-1}] \\
    & = \sum_k p T_k \frac{1}{T_k^2} \leq \gamma^3 p^2|J|h(S)/(1-\delta)\alpha_0 f|S| \\
    &\leq \gamma^3 p^2|J|n(S)/(1-\delta)\alpha_0 f|S| := v.
\end{align*}
Let the deviation $t = \epsilon p|J| = \epsilon \phi p n(S)$, then we have $b = \frac{ \gamma^3 t}{(1-\delta)\alpha_0 \epsilon \phi f |S|}$ and $v = \frac{\gamma^3t^2}{(1-\delta)\alpha_0 \epsilon^2\phi  f |S|}$. Using the minimum value of $|S|$ is $\ell/\phi$ and applying Theorem \ref{thm:martingale} to $-X_{|J|}$, we have
\begin{align*}
    P[D(S) < (1-\epsilon)pn^{\mathcal{H}}(S)] \leq \exp({-\frac{\epsilon t}{2b(1 + \epsilon/3)}}) \leq \exp({-\lambda}).
\end{align*}
This concludes the proof.

\end{proof}

\begin{proof}[Proof of Lemma~\ref{lem:good}]
We prove the lemma by induction.
For epoch $e = 1$, it is trivial to just set $T_1 = f/p \tilde h_1$ by Assumption~\ref{ass:main1}.2.
Now we assume all epochs are good until epoch $e-1$ ($e \geq 2$), we will show epoch $e$ is good.

Let $S_0 = \{u:r_1\leq u \leq r_2\}$ be the window that will be used to determine $T_e$, i.e., $r_1 =\max{(0,(e-2)R-\kappa)}+1$ and $r_2 = (e-1)R-\kappa$. Let $\mathcal{C}(S_0)$ be the segment of $\mathcal{C}$ containing PoS blocks with timestamps in $S_0$, let $\B$ be all PoW blocks included in $\mathcal{C}(S_0)$, and $d$ be the total difficulty of all PoW blocks in $\B$. Then we have the following facts:
\begin{itemize}
    \item {\bf Fact 1.} For any PoW block $B \in \B$, $B$ is mined after $r_1 - 4\kappa - \Delta$. Indeed by recency condition, $B$ must refer to a confirmed PoS block $B_s$ with timestamp at least $r_1 - sl_{\rm re}$. By $\exists$CQ, the last honest ancestor block of $B_s$ has timestamp at least $r_1 - sl_{\rm re} - \kappa$. So $B$ must be mined after $r_1 - 4\kappa - \Delta$.
    \item {\bf Fact 2.} For any PoW block $B \in \B$, $B$ is mined before $r_2 + \kappa$. Indeed the PoS block (denoted as $B_s$) including $B$ has timestamp at most $r_2$, and again by $\exists$CQ, the first honest descendant block of $B_s$ has timestamp at most $r_2 + \kappa$. So $B$ must be mined before $r_2 + \kappa$.
    \item {\bf Fact 3.} If a PoW block $B$ is mined by an honest miner after $r_1$ and before $r_2 - 3\kappa -\Delta$, then $B \in \B$. Indeed, by $\exists$CQ, the last honest block in $\mathcal{C}(S_0)$ has timestamp at least $r_2 - \kappa$. Hence by Lemma~\ref{lem:fresh}, all honest PoW blocks mined after $r_1$ and before $r_2 - \kappa - r_{\rm wait}$ will be included into a PoS block in $\mathcal{C}(S_0)$.  
\end{itemize} 

Let $S_1 = \{u: r_1-(4\kappa + \Delta) \leq u \leq r_2 + \kappa \}$, $S_2 = \{u: r_1 \leq u \leq r_2 - (3\kappa +\Delta)\}$, and $J$ be the set of adversary queries associated with $\B$ in $S_1$. Then by Fact 1 and Fact 2, we have all PoW blocks in $\B$ are mined in $S_1$; by Fact 3, we have all PoW blocks mined by honest nodes in $S_2$ are in $\B$. 
Hence, $D(S_2) \leq D_{e-1} \leq D(S_1) + A(J)$.
By Proposition~\ref{prop:good_slot}, for each query in $S_1$ made by an honest node at slot $r$ in epoch $e$, the target $T$ must be $T_e$, so we have $(1-\delta)\alpha_0 f/\gamma^2 \leq p h_r T \leq (1+\delta)\gamma^2 f$.
For each query in $J$ made by the adversary at slot $r$ in epoch $e$, the target $T$ may be either $T_{e-1}$ or $T_e$, still we have $(1-\delta)\alpha_0 f/\gamma^2 \leq p h_r T \leq (1+\delta)\gamma^2 f$ in both cases (under a $(\gamma,2R)$-respecting environment). By the fact that $\alpha_0 n_r \leq h_r \leq n_r$, we have $(1-\delta)\alpha_0 f/\gamma^2 \leq p n_r T \leq (1+\delta)\gamma^2 f/\alpha_0$.

Under a typical execution, we have
\begin{equation*}
    D(S_2) > (1-\epsilon)ph(S_2),
\end{equation*}
and
\begin{align*}
    D(S_1) + A(J) &< (1+\epsilon) p(h(S_1)+|J|) \\
    &= (1+\epsilon) pn(S_1) \leq (1+\epsilon) ph(S_1)/\alpha_0.
\end{align*}
Therefore, by Proposition~\ref{prop:envbounds}, we have
\begin{equation*}
    (1-\epsilon) ph^{(e)}|S_2|/\gamma < D_{e-1}^{\rm total} <  (1+\epsilon) \gamma ph^{(e)}|S_1|/\alpha_0.
\end{equation*}

By the difficulty adjustment rule, we have $T_e = f|S_0|/D_{e-1}^{\rm total}$. Then $T_e$ can be bounded as follows:
\begin{itemize}
    \item {\bf Lower bound:}
    \begin{align*}
        ph^{(e)}T_e &\geq \frac{|S_0|}{(1+\epsilon)|S_1|} \frac{\alpha_0 f}{\gamma} \geq \frac{(1-\epsilon)|S_0|}{|S_0| + 5\kappa+\Delta}\frac{\alpha_0 f}{\gamma} \\
        & \geq \frac{(1-\epsilon)(R-\kappa)}{R+ 4\kappa+\Delta}\frac{\alpha_0 f}{\gamma} \overset{(C_1)}{\geq} \frac{(1-\epsilon)(R-\epsilon R)}{R+ \epsilon R}\frac{\alpha_0 f}{\gamma} \\
        & \geq (1-\epsilon)^3 \frac{\alpha_0 f}{\gamma} \geq (1-4\epsilon) \frac{\alpha_0 f}{\gamma} \overset{(C_2)}{\geq} (1-\delta)\frac{\alpha_0 f}{\gamma}.
    \end{align*}
    \item {\bf Upper bound:}
    \begin{align*}
        ph^{(e)}T_e &\leq \frac{|S_0|}{(1-\epsilon)|S_2|} \gamma f = \frac{|S_0|}{(1-\epsilon)(|S_0| - 3\kappa-\Delta)}\gamma f \\
        &\leq \frac{R-\kappa}{(1-\epsilon)(R- 4\kappa-\Delta)}\gamma f  \overset{(C_1)}{\leq} \frac{R}{(1-\epsilon)(R- \epsilon R)}\gamma f \\
        &= \frac{1}{(1-\epsilon)^2} \gamma f \leq \frac{1}{1-2\epsilon} \gamma f \leq (1+4\epsilon) \gamma f \overset{(C_2)}{\leq} (1+\delta)\gamma f.
    \end{align*}
\end{itemize}
This concludes the proof.
\end{proof}

\begin{proof}[Proof sketch of Theorem~\ref{thm:re_main}]
When moving from the single-epoch setting to the multiple-epoch setting, two new aspects need to be considered.
\begin{itemize}
    \item {\bf Virtual stake distribution updates.} For epochs $e=1,2$, the virtual stake has the same distribution as the initial stake. Since we assume the initial stake has honest majority (Assumption~\ref{ass:main1}.1), by Theorem~\ref{thm:single}, CP and $\exists$CQ are guaranteed in epochs $1\&2$. For epoch $e \geq 3$, the virtual stake of a node $h$ composes of two parts, the actual stake recorded on the blockchain up to the last block of the epoch $e-2$ (by Ouroboros~\cite{kiayias2017ouroboros}) and the work stake decided by the amount of work $h$ has contributed in epoch $e-2$. Let $S_s$ and $S_w$ be the total actual stake and total work stake in epoch $e$. Recall that we set $\omega S_s = (1-\omega)S_w$ in \minotaurnosp, i.e., the total virtual stake is $S_v = S_s/(1-\omega)$. Denote that the adversary controls $\beta_s S_s$ actual stake and let $\beta_w$ be maximum fraction of adversarial mining power in epoch $e-2$, i.e., $\beta_w = 1 - \min_{r \in [(e-3)R+1,(e-2)R]}{(h_r/n_r)}$. By Assumption~\ref{ass:main2}.1, we have $\omega\beta_w + (1-\omega)\beta_s \leq 1/2-2\sigma$. And by Theorem~\ref{thm:fairness}, honest nodes control at least $(1-\sigma)(1-\beta_w)S_w$ work stake for epoch $e$. Therefore, the virtual stake controlled by the adversary is at most 
    \begin{align*}
          &\;\;\;\; \beta_s S_s + (1-(1-\sigma)(1-\beta_w))S_w \\
          &= ((1-\omega) \beta_s + \omega(\beta_w +\sigma(1-\beta_w)))) S_v \\
          &\leq (1/2-\sigma)S_v.
    \end{align*}
    Therefore, by an induction argument, we can guarantee an honest majority in the virtual stake for all epochs. The analysis critically relies on the fact that the CP property is immutable: specifically, when all honest parties agree on a common prefix $\mathcal{C}^{(t)}$ at some slot $t$ and, as {\sf maxvalid-mc} can only revise the last $k$ blocks of a currently adopted chain, $\mathcal{C}^{(t)}$ will be a prefix of all future chains held by the honest parties. Check the proof of Theorem 5.3 in~\cite{kiayias2017ouroboros} for details.
    \item {\bf Randomness updates.} Every epoch needs new public randomness to be used for sampling slot leaders from the above virtual stake distribution. In Ouroboros~\cite{kiayias2017ouroboros}, elected slot leaders (one per slot) from epoch $e-1$ runs a publicly verifiable secret sharing (PVSS) protocol to generate the randomness for epoch $e$. The core idea is the following: given that we have guaranteed that an honest majority among elected leaders in epoch $e$ will hold with very high probability, we have that the PVSS protocol suitably simulates a beacon with the relaxation that the output may become known to the adversary before it is known to the honest parties. However, as long as the distribution of virtual stake is determined prior to this leakage, the sampling of leaders in epoch $e$ will still be unbiased.
\end{itemize}

\revreplace[Request~\ref{req:qualityandgrowth}\label{ans:qualityandgrowth}]{At last, it is not hard to see that, the CP property is equivalent to the persistence of the ledger, while CP together with $\exists$CQ implies liveness.}{At last, by our definition, the CP property with parameter $\ell_{cp} = \kappa$ is equivalent to the persistence of the ledger (when the confirmed ledger is defined as $\mathcal{C}^{\lceil \kappa}$). Meanwhile, the CP property with parameter $\ell_{cp} = \kappa$, together with the $\exists$CQ property with parameter $\ell_{cq} = \kappa$, implies liveness with parameter $u = 2\kappa$. Indeed, for a chain $\mathcal{C}$ held by an honest node at slot $r$, there must be at least one honest block $B^*$ in $\mathcal{C}^{\lceil \kappa}$ with timestamp in $[r-2\kappa,r-\kappa]$ (by $\exists$CQ), then any transaction that appears before $r-2\kappa$ should be included by either $B^*$ or its ancestors, thus the transaction should be confirmed (by CQ/persistence).}
\end{proof}

\section{Details on the Ouroboros Protocol Family}\label{app:ouroboros}
We give a summary of the different PoS protocols we explicitly base
our generic hybridization construction on in the main part of the paper,
Ouroboros Classic~\cite{kiayias2017ouroboros}, Praos~\cite{david2018ouroboros},
and Genesis~\cite{badertscher2018ouroboros}.
In order of this sequence, each version of the protocol gives stronger security guarantees.

We first give quick summaries of the respective protocol guarantees and their
underlying assumptions.
Finally, in Appendix~\ref{app:praos}, we give a more detailed description of
Ouroboros Praos as the reference protocol; and sketch how the other variants differ from it.

\paragraph{Ouroboros Classic.}
Classic is secure against a minority of adversarially controlled stake under the following
assumptions:
\begin{description}
 \item{Network.} The communication network is synchronous.
 \item{Corruption.} The adversary is `moderately' adaptive (participant corruption only
   takes effect after a certain delay).
 \item{Stake shift.} There is an upper bound on the stake shift, i.e., the stake
   distribution among the stake holders does not change too fast.
 \item{Offline tolerance.} The protocol participants only go offline for short periods
   of time.
\end{description}

\paragraph{Ouroboros Praos assumptions:}
Praos is secure against a minority of adversarially controlled stake under the following
assumptions:
\begin{description}
 \item{Network.} The communication network  is semi-synchronous, i.e., that the network
   delay is bounded bu some delay $\Delta$ not known to the participants.
 \item{Corruption.} The adversary is fully adaptive.
 \item{Stake shift.} As in Classic.
\end{description}

\paragraph{Ouroboros Genesis assumptions:}
Classic is secure against a minority of adversarially controlled stake among all
participants who are active in the system---under the following assumptions:
\begin{description}
 \item{Network.} As in Praos.
 \item{Corruption.} As in Praos.
 \item{Offline tolerance.} The protocol participants can join later during any stage of
   the protocol, or go offline for extended periods of time during participation.
\end{description}

\subsection{Ouroboros Praos}\label{app:praos}
The protocol proceeds in epochs of $R$ slots, each slot representing a given `unit of time',
say, $1$ second of the protocol run-time. For ease of exposition, let the genesis block
represent epoch $0$ of the protocol. We now describe how the protocol operates per epoch $e>0$.

\paragraph{Slot-leader election.}
During each slot $\slot_j$, a slot-leader election is held among the stakeholders, and a
winning stakeholder is allowed to publish a new block associated with this slot.

The slot-leader election during epoch $e$ is based on the stake distribution at the end of
epoch $\max(0,e-2)$, i.e., the stake distribution that results after the processing of the
last block of epoch $e-2$ (or the stake distribution from the genesis block).

Let $\alpha_i$ be the relative stake held by stake holder $\party_i$ (holding $\alpha_i\cdot S$
of the total stake $S$) at the end of epoch $\max(e-2,0)$. Per slot, the probability $p_i$ for
stakeholder $\party_i$ to be a block leader is defined as
\begin{equation}\label{eq:praos:phi}
 p_i=\phi_f(\alpha_i)\defeq 1-(1-f)^{\alpha_i}
\end{equation}
for some appropriate \emph{active-slots coefficient} $f$ (the probability that, among the total
stake, at least one slot leader is elected during any given slot).

Slot leadership is pseudo-randomly assigned based on the \emph{epoch nonce} $\eta$, a seed
calculated in epoch $e-1$ (as described further below). To become a slot leader for slot $\slot_m$,
the stakeholder $\party_i$ evaluates a verifiable random function (VRF) (bound to a public key registered
by $\party_i$) $(y,\pi)=\VRF_i(\eta,\slot_m)$. Slot leadership is satisfied iff
\[
 y\stackrel{!}{<}T\defeq 2^{\vrflengthout}\cdot\phi_f(\alpha_i),
\]
where $\VRF_i$ produces outputs of $\vrflengthout$ bits, and $T$ is a threshold to enforce the
desired probabilities.

\paragraph{Block production.}
Besides a hash of its predecessor and the payload, a block contains
\begin{itemize}
    \item the slot number $m$;
    \item the above proof of leadership $(y,\pi)$ (such that  $y<T$);
    \item an additional, independent, VRF output $(y_\rho,\pi_\rho)$ contributing to the epoch-nonce
      generation; and
    \item a signature on $(m,(y,\pi),(y_\rho,\pi_\rho))$ by $\party_i$ of a key-evolving signature
      scheme (KES).\footnote{The stakeholders update the private keys of their KES instance after
        every slot.}
\end{itemize}

\paragraph{Block settlement.}
A block is considered settled if it sits at least $k$ blocks deep in a node's main chain where $k$
is the \emph{prefix parameter} of the protocol.

\paragraph{Chain selection rule: {\sf maxvalid-mc}.}
Upon the arrival of a new block, each node chooses, as their main chain, the longest
chain that does not fork from the previous main chain by more than $k$ blocks where
$k$ is the prefix parameter.

\paragraph{Epoch-nonce generation.}
The epoch nonce $\eta$ for epoch $e>1$ is computed as a hash of all VRF outputs $y_\rho$ included
in the blocks of the main chain up to $2/3$ of epoch $e-1$.
The epoch nonce for epoch $e=1$ is included in the genesis block.

\subsection{Differences in Classic and Genesis}

\paragraph{Classic.}
The main difference to Praos is that the epoch randomness is obtained from a coin-flip protocol
based on a publicly verifiable secret sharing protocol---with the disadvantage that the leader
schedule is public and thus allows for adaptive corruption.

\paragraph{Genesis}
Genesis differs from Praos by applying a different chain selection rule to allow protocol
participants to be offline for extended periods of time. The modified chain-selection rule
{\sf maxvalid-bg} works as follows:
A longer chain is adopted if
\begin{itemize}
    \item it forks from the current main chain by at most $k$ blocks; or
    \item if forks by more than $k$ blocks but contains more blocks in the $s$ slots following
      the last common block of the main chain and the longer chain.
\end{itemize}

\section{Analysis tightness}
\label{app:tight}


The analysis of our main security theorem, Theorem~\ref{thm:main1}, draws from \cite{kiayias2017ouroboros}.
A more careful and improved analysis appears in \cite{bkmqr2020}, from which we may obtain the refined bound
$1 - RL(e^{-\Omega(\delta^3\kappa)}+e^{-\lambda})$ for the statement of Theorem~\ref{thm:main1}.
In this section we analyze a private chain attack and show a corresponding lower bound.
It reveals that $\delta^2\kappa$ needs to be bounded below by a constant and, consequently,
the dependency on $\delta^3$ cannot be improved much further.

\begin{prop}
	Common prefix does not hold against a
	$(1/2-\delta,m,\omega)$-bounded adversary, for $\delta\le1/\sqrt{\ell_{\rm cp}}$ and
	sufficiently large $\ell_{\rm cp}$.
\end{prop}
\begin{proof}
	The adversary follows the protocol using his mining power to produce PoW
	blocks. The attack begins when the virtual stake of the adversary is
	$\omega\beta_w+(1-\omega)\beta_s\ge\frac12-\delta$.
	
	Define the random variable $X_i$ taking values in $\{-1,0,1\}$ according to who produced 
	blocks of height $i$. If it was the adversary only, then $X_i=1$; if both
	the adversary and honest parties, then $X_i=0$; if honest parties only, then
	$X_i=-1$.
	Let $p=1-(1-f)^{1/2+\delta}$ and $q=1-(1-f)^{1/2-\delta}$.
	We have
	\begin{align*}
		&\Pr[X_i=-1]=p(1-q)/f,\quad \\
		&\Pr[X_i=1]=q(1-p)/f,\quad\text{and}\quad \\
		&\Pr[X_i=0]=pq/f.
	\end{align*}
	Note that the adversary will create a fork of length $k$, if $X_1+\cdots+X_k\ge0$.
	We will use the Berry-Esseen bound to lower bound the probability he
	succeeds. To that end we compute
	\begin{align*}
		\mu&
			=\E[X_i]=(q-p)/f
			=[(1-f)^{1/2+\delta}-(1-f)^{1/2-\delta}]/f
		    \\&
		    =[(1-f)^{1+\delta}-(1-f)^{1-\delta}]/[f(1-f)^{1/2}]
		    \\&
		    >[(1-f)^{1+\delta}-(1-f)^{1-\delta}]/f
			>-2\delta
		,\\
		\sigma^2&
			=\E[X_i^2]-\mu^2
			=1-pq/f-\mu^2
			>1-f-4\delta^2
		,\\
		\rho&
			=\E[|X_i-\mu|^3]
			<(1+2\delta)^3
	.\end{align*}
	Observe that for $\delta\le1/3$ and $f\le1/9$.

	Let $F_k=(X_1+\cdots+X_k-\mu k)/\sqrt{k\sigma^2}$.
	By the Berry-Essen Theorem 
	we have that, for $\delta\le1/\sqrt k$,
	\begin{align*}
		\Pr\Bigl[F_k\ge-\mu k/\sqrt{k\sigma^2}\Bigr]
		&\ge\Phi^c\bigl(3)-O\bigl(k^{-1/2}\bigr)
	.\end{align*}
\end{proof}

Although it might be possible to strengthen the above attack and analysis to obtain slight improvements,
it remains an open question whether the exponential drop in the probability of security failure 
is in the order of $\delta^3\kappa$ or $\delta^2\kappa$.

\end{document}